\documentclass[letterpaper, 10 pt, conference]{ieeeconf}  
\usepackage{graphicx}
\usepackage{setspace}
\usepackage{graphics} 
\usepackage{epsfig} 
\usepackage{epstopdf}
\usepackage[abs]{overpic}
\usepackage{algorithmic}
\usepackage{amsmath}
\usepackage{amsfonts}

\usepackage{tikz}
\usepackage{epsfig} 
\usepackage{theorem}
\newtheorem{theorem}{\bf{Theorem}}[section]

\newtheorem{cor}[theorem]{Corollary}
\newtheorem{lem}[theorem]{Lemma}

\theoremstyle{plain}

\graphicspath{{figs/}}

\newcounter{algno} 
\newenvironment{alg}{\refstepcounter{algno}\small\tabular}{\endtabular}

\title{\LARGE \bf
Formation of Robust Multi-Agent Networks Through Self-Organizing Random Regular Graphs
}

\author{A. Yasin Yaz{\i}c{\i}o\u{g}lu, Magnus Egerstedt, and Jeff S. Shamma\\
\thanks{ This work was supported by AFOSR/MURI \#FA9550-10-1-0573 and ONR Project \#N00014-09-1-0751.\newline \indent  A. Yasin Yaz{\i}c{\i}o\u{g}lu is with the Laboratory for Information and Decision Systems, Massachusetts Institute of Technology, {\tt yasiny@mit.edu}. \newline \indent
Magnus Egerstedt is with the School of Electrical and Computer Engineering, Georgia Institute of Technology, {\tt magnus@gatech.edu}.\newline \indent
Jeff S. Shamma is with the School of Electrical and Computer Engineering, Georgia Institute of Technology, {\tt shamma@gatech.edu}, and with King Abdullah University of Science and Technology (KAUST), {\tt jeff.shamma@kaust.edu.sa.}}}

\IEEEoverridecommandlockouts
\begin{document}
\maketitle
\begin{abstract}
Multi-agent networks are often modeled as interaction graphs, where the nodes represent the agents and the edges denote some direct interactions. The robustness of a multi-agent network to perturbations such as failures, noise, or malicious attacks largely depends on the corresponding graph. In many applications, networks are desired to have well-connected interaction graphs with relatively small number of links. One family of such graphs is the random regular graphs. In this paper, we present a decentralized scheme for transforming any connected interaction graph with a possibly non-integer average degree of $k$ into a connected random $m$-regular graph for some $m\in [k, k+2]$. Accordingly, the agents improve the robustness of the network with a minimal change in the overall sparsity by optimizing the graph connectivity through the proposed local operations.

\end{abstract}
\section{Introduction}
Multi-agent networks have been used to characterize a large number of natural and engineered systems such as biological systems, financial networks, social networks, communication systems, transportation systems, power grids, and robotic swarms. Multi-agent networks are often represented via their interaction graphs, where the nodes correspond to the agents and the edges exist between the agents having some direct interaction. Various system properties such as the robustness, mixing time, or controllability of a network greatly depend on the structure of the corresponding graph (e.g., \cite{Dekker04,Jamakovic07, Olfati05Ultra,Rahmani09,Yasin12CDC}). Therefore, graph theoretic
analysis of networked systems has received a considerable
amount of attention during the last decade (e.g., \cite{Mesbahi10, Barrat08}). 

In many applications, multi-agent networks face various perturbations such as component failures, noise, or malicious attacks. One of the fundamental measures that capture the robustness of networks to such perturbations is the connectivity of the corresponding graph. A graph is said to be $k$-node (or -edge) connected if at least $k$ nodes (or edges) should be removed to render the graph disconnected. In general, graphs with higher connectivity have higher robustness to the targeted removal of their components \cite{Dekker04, Jamakovic07}.  For many networks, well-connected interaction graphs provide robustness to not only the structural failures but also the functional perturbations such as noise (e.g., \cite{Young10, Abbas12}). Connectivity can also be quantified via some spectral measures such as the algebraic connectivity \cite{Fiedler73} or the Kirchhoff index \cite{Klein93}. An arguably richer measure of connectivity is the edge (or node) expansion ratio (e.g., \cite{Pinsker73,Alon86}). If the expansion ratio of a graph is small, then it  is possible to disconnect a large set of nodes by removing only a small number of edges (or nodes). Graphs with high expansion ratios are called expanders. A detailed overview of expanders and their numerous applications can be found in \cite{Hoory06} and the references therein. 

The connectivity of a network can be improved by adding more edges to the graph. However, each edge stands for some communications, sensing, or a physical link between the corresponding agents. Hence, more edges require more resources. Moreover, too many edges may lead to higher vulnerability to the cascading failures such as epidemics (e.g., \cite{Blume11,Ganesh05}). Consequently, having a small number of edges (i.e., sparsity) is desired in many applications. One family of well-connected yet sparse graphs is the random regular graphs. A graph is called an $m$-regular graph if  the number of edges incident to each node (the degree) is equal to $m$. A random $m$-regular graph of order $n$ is a graph that is selected uniformly at random from the set of all $m$-regular graphs with $n$ nodes. As $n$ goes to infinity, almost every $m$-regular graph is an expander for any $m \geq 3$ \cite{Alon86,Friedman03}. 

In this paper, we present a decentralized scheme to transform any connected interaction graph into a connected random regular graph with a similar number of edges as the initial graph. Some preliminary results of this work were presented in \cite{Yasin13CDC} and \cite{Yasin14CDC}, where any initial graph with an integer average degree, $k \in \mathbb{N}$, was transformed into a random $k$-regular graph. In this paper, we extend our earlier works to obtain a method that is applicable to the most generic case, i.e. $k \in \mathbb{R}$. The proposed method transforms any graph with a possibly non-integer average degree of $k$ into a random $m$-regular graph for some $m \in [k, k+2]$. As such, the graph becomes well-connected, regardless of the network size.







The organization of this paper is as follows: Section \ref{relwork} presents some related work. Section \ref{prelim} provides some graph theory preliminaries. Section \ref{scheme} formulates the problem and presents the proposed solution. Section \ref{distalg} provides a distributed implementation and an analysis of the resulting dynamics. Section \ref{sims} provides some simulation results to demonstrate the performance of the proposed solution. Finally, Section \ref{conclusion} concludes the paper.




%

\section{Related Work}
\label{relwork}
  
In some applications, the robustness of a graph is related to the centrality measures such as the degree, betweenness, closeness, and eigenvector centralities. Loosely speaking, the centrality measures capture the relative importances of the nodes in a graph. Detailed reviews on the centrality measures and their applications can be found in \cite{Borgatti05,Koschutzki05} and the references therein. Typically, the perturbations applied to the nodes with higher centrality scores have a stronger impact on the overall system (e.g., \cite{Acemoglu12,Jeong01,Holme02,Wang10}). Hence, graphs with unbalanced centrality distributions are usually vulnerable to such worst-case perturbations.

In the literature, there are many works related to the construction of robust interaction graphs. Some of these works consider how a robustness measure can be improved via some modifications to the graph topology. For instance, such improvement can be achieved by rewiring a small percentage of the existing edges (e.g., \cite{Schneider11}) or adding a small number of edges to the graph (e.g., \cite{Beygelzimer05,Abbas12}). 

Another group of studies consider the explicit construction of expanders. Expanders can be constructed via graph operations such as zig-zag product (e.g., \cite{Reingold02,Capalbo02}), or derandomized graph squaring \cite{Rozenman05}. Furthermore, for any $m \in \mathbb{N}$ such that $m-1$ is a prime power, an explicit algebraic construction method for a family of $m$-regular expanders, i.e. Ramanujan graphs, was presented in \cite{Morgenstern94}. In \cite{Olfati07}, Watts-Strogatz small-word networks are transformed into quasi Ramanujan graphs by rewiring some of the edges.

Expanders can also be built as random $m$-regular graphs. A detailed survey of the various models of random regular graphs as well as their properties can be found in \cite{Wormald99} and the references therein. As $n$ goes to infinity, almost every $m$-regular graph has an algebraic connectivity arbitrarily close to $m-2\sqrt{m-1}$ for $m \geq 3$ \cite{Alon86,Friedman03}. For regular graphs, high algebraic connectivity implies high expansion ratios (e.g., \cite{Mohar89}). Hence, for any fixed $m \geq 3$, a random $m$-regular graph has the algebraic connectivity and expansion ratios bounded away from zero, even if the network size is arbitrarily large.

A random $m$-regular graph with $n$ nodes can be constructed by generating $m$ copies for each node, picking a uniform random perfect matching on the $nm$ copies, and connecting any two nodes if the matching contains an edge between their copies (e.g., \cite{Bollobas80,Steger99}). In \cite{Law03}, the authors present a distributed scheme for incrementally building random $2m$-regular multi-graphs with $m$ Hamiltonian cycles. Alternatively, some graph processes may be designed to transform an initial $m$-regular graph into a random $m$-regular graph by inducing a Markov chain with a uniform limiting distribution over the set of $m$-regular graphs  (e.g., \cite{Jerrum90, Mahlmann05}). The method in this paper is also based on designing a graph process with a uniform limiting distribution over the set of $m$-regular graphs. Compared to the similar works in the literature, the proposed scheme is applicable to the most generic case, and it is decentralized. The initial graph is not required to satisfy some strong properties such as being regular or having an integer average degree. Furthermore, the global transformation is achieved via only some local graph operations.

\section{ Preliminaries}
\label{prelim}


 
An undirected graph, $\mathcal{G}=(V,E)$, consists of a set of nodes, $V$, and a set of edges, $E$, given by unordered pairs of nodes. A graph is connected if there exists a path between any pair of nodes. A path is a sequence of nodes such that an edge exists between any two consecutive nodes in  the sequence. Any two nodes are said to be adjacent if an edge exists between them. We refer to the set of nodes adjacent to any node, $i \in V$, as its neighborhood, $\mathcal{N}_i$, defined as
\begin{equation}
\label{neigh}
\mathcal{N}_i= \{j \mid (i,j) \in E\}.
\end{equation}  

For any node $i$, the number of nodes in its neighborhood is called its degree, $d_i$, i.e., 
\begin{equation}
\label{deg}
d_i=|\mathcal{N}_i|,
\end{equation}
where $|\mathcal{N}_i|$ denotes the cardinality of $\mathcal{N}_i$. For any graph $\mathcal{G}$, we use $\delta(\mathcal{G})$, $\Delta(\mathcal{G})$ and $\bar{d}(\mathcal{G})$ to denote the minimum, the maximum, and the average degrees, respectively. We refer to the difference of the maximum and the minimum node degrees in a graph as the degree range, $f(\mathcal{G})$, i.e. 
\begin{equation}
\label{feq}
f(\mathcal{G})=\Delta(\mathcal{G}) -\delta(\mathcal{G}).
\end{equation}

For any undirected graph, $\mathcal{G}=(V,E)$, the graph Laplacian, $L$, is a symmetric matrix whose entries are given as

\begin{equation}
\label{Lapl}
[L]_{ij}=\left\{\begin{array}{ll}\;d_i&\mbox{ if } 
i=j,\\-1&\mbox{ if  }(i,j)\in E, \\ \;0 & \mbox{ otherwise. }\end{array}\right.
\end{equation}
The second-smallest eigenvalue of the graph Laplacian is known as the algebraic connectivity of the graph, $\alpha(\mathcal{G})$. 



%
%

Local graph transformations can be represented using the framework of graph grammars (e.g., \cite{Klavins06}). A grammar, $\Phi$, is a set of rules, where each rule is defined as a label-dependent graph transformation. Each rule is represented as an ordered pair of labeled graphs, $r=(g_l, g_r)$, where the labels represent the node states. Graph grammars operate on labeled graphs. A labeled graph, $\mathcal{G}=(V,E,l)$, consists of a node set, $V$, an edge set, $E$, and a labeling function, $l: V \mapsto \Sigma$, where $\Sigma$ is the set of feasible node labels. A rule is said to be applicable to a labeled graph, $\mathcal{G}=(V,E,l)$, if $\mathcal{G}$ has a subgraph isomorphic to $g_l$, i.e. if there is a bijection, which preserves node labels and edges, between $g_l$ and a subgraph of $\mathcal{G}$. A rule, $r=(g_l, g_r)$, transforms any graph isomorphic to $g_l$ to a graph isomorphic to $g_r$. A labeled initial graph, $\mathcal{G}(0)$, along with a grammar, $\Phi$, defines a non-deterministic system represented as $(\mathcal{G}(0),\Phi )$.

\section{Decentralized Formation of Random Regular Graphs}
\label{scheme}

\subsection{Problem Formulation}
Motivated by the connectivity properties of almost every $m$-regular graph for $m\geq3$, this paper is focused on the following problem: Assume that a multi-agent network is initialized with an arbitrary connected interaction graph. How can the agents reconfigure their links locally in a decentralized fashion such that the resulting dynamics transform the interaction graph into a connected random regular graph with a similar sparsity as the initial graph? Such a transformation is illustrated in Fig. \ref{first_obj}.





\begin{figure}[htb]
\includegraphics[scale=0.25]{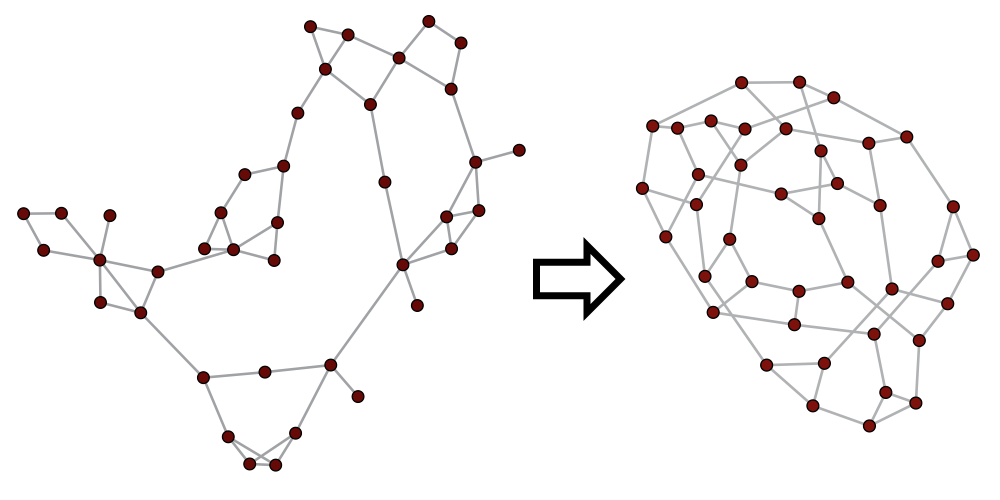}
\put(-180,-5){(a)}
\put(-60,-5){(b)}
\caption{A fragile interaction graph (a) with an average degree 2.8 is transformed into a random 3-regular graph (b). The graphs have similar sparsities, whereas the structure in (b) has a significantly improved connectivity. }
\label{first_obj}
\end{figure}


In order to transform the interaction graph into a random regular graph, the agents need to pursue three global objectives: 1) balance the degree distribution, 2) randomize the links, and 3) drive the average degree to an integer close to its initial value. Furthermore, the agents should ensure that the graph remains connected as they modify their links. This paper presents a set of locally applicable graph reconfiguration rules to pursue all of these global objectives in parallel. 
In the remainder of this section, we will incrementally build the proposed scheme using the framework of graph grammars. 

\subsection{Degree Balancing}

As the first step towards building random regular graphs, we present a local rule for balancing the degree distribution in a network while maintaining the connectivity and the total number of edges. The degree balancing task can be considered as a quantized consensus problem (e.g., \cite{Kashyap07,Nedic09}), where each local update needs to be realized via some local changes in the structure of the graph. In this part, we design a single-rule grammar, $\Phi_{R}$, for balancing the degree distribution. In $\Phi_{R}$, each node is labeled with its degree, i.e. 

\begin{equation}
\label{label_degree}
l(i)=d_i, \text{\;\;}\forall i\in V.
\end{equation}
The proposed grammar, $\Phi_R$, is defined as 
$$ \includegraphics[scale=0.40]{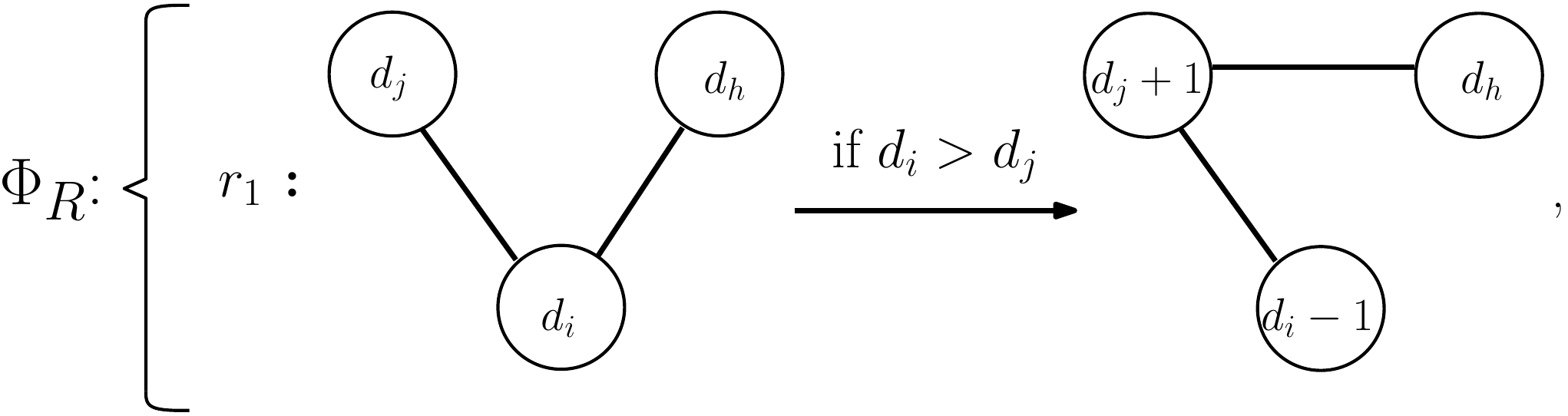}$$
where $d_i$, $d_j$, and $d_h$ denote the degrees of the corresponding nodes. 

In accordance with $\Phi_R$, nodes behave as follows: Let $i$ and $j$ be any pair of adjacent nodes. If $d_j<d_i$, then a new link is formed between $j$ and an arbitrary neighbor of $i$, say $h$, that is not currently linked with $j$. At the same time, the link between $i$ and $h$ is terminated to maintain the overall sparsity. Note that if $d_j<d_i$, then $i$ has at least one such exclusive neighbor, $h$. Furthermore, if a connected graph is not regular then there always exists at least one pair of adjacent nodes, $i$ and $j$, such that $d_i>d_j$. Hence, any connected graph is stationary (i.e., $\Phi_R$ is not applicable anywhere on the graph) if and only if it is a regular graph. Moreover, any concurrent application of $\Phi_R$ on a connected graph, $\mathcal {G}$, maintains the connectivity and the average degree of $\mathcal{G}$. 


A detailed analysis of the dynamics induced by $\Phi_R$ can be found in \cite{Yasin13CDC}, where it is shown that the degree range, $f(\mathcal{G})$, monotonically decreases and converges to its minimum feasible value under $\Phi_R$. As such, the graph converges to a regular graph (i.e., $f(\mathcal{G})=0$) if the average degree is an integer. Otherwise, $f(\mathcal{G})$ converges to 1. For the sake of brevity, we skip the details of this analysis and only provide its main result in Theorem \ref{degreg}.


 \begin{theorem}\label{degreg}
Let $\mathcal{G}(0)$ be a connected graph and let $\tau =\{\mathcal{G}(0), \mathcal{G}(1), \hdots\}$ be a feasible trajectory of $(\mathcal{G}(0), \Phi_R)$. Then, $\tau_f=\{f(\mathcal{G}(0)), f(\mathcal{G}(1)), \hdots\}$, almost surely converges to an integer, $\tau_f^*$, such that
\begin{equation}
\label{taufconv}
\tau_f^*=\left\{\begin{array}{ll}0&\mbox{ if } 
\bar{d}(\mathcal{G}(0))\in \mathbb{N} \\ 1&\mbox{ otherwise, }\end{array}\right.
\end{equation}
where $\bar{d}(\mathcal{G}(0))$ denotes the average degree of $\mathcal{G}(0)$.
\end{theorem}


The grammar $\Phi_R$ transforms any connected graph with an integer average degree, $k  \in \mathbb{N} $, into a connected $k$-regular graph. Note that, although almost every $k$-regular graph is an expander for $k\geq3$, $\Phi_R$ may still result in a rare configuration with an arbitrarily small expansion rate. For instance, both of the 3-regular graphs in Fig. \ref{goodandbad} are stationary under $\Phi_R$ whereas the graph in Fig. \ref{goodandbad}a can have half of the network disconnected due the failure of a single edge or a single node. The probability of converging to an undesired equilibrium such as the graph in Fig. \ref{goodandbad}a depends on the initial graph.  In the next part, we will build on $\Phi_{R}$ to avoid such undesired outcomes.

\begin{figure}[htb]
\begin{center}
\includegraphics[trim = 5mm 20mm 0mm 4mm,scale=0.27]{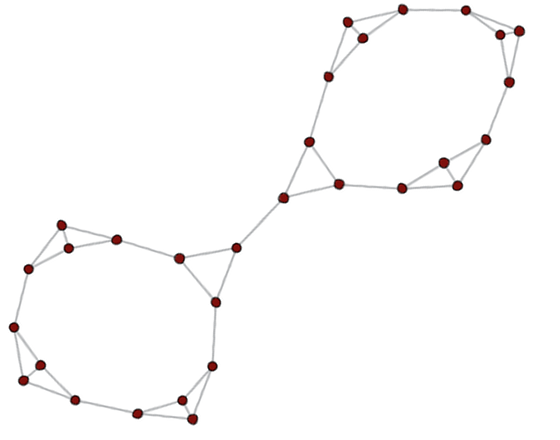}
\includegraphics[trim = 0mm 0mm 0mm 4mm,scale=0.55]{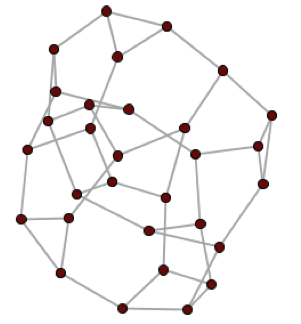}
\put(-170,-20){\small (a)}
\put(-50,-20){\small (b)}
\caption{A poorly-connected 3-regular graph (a) and a well-connected 3-regular graph (b).  }
\label{goodandbad}
\end{center}
\end{figure}

\subsection{Link Randomization}
In order to ensure that the interaction graph does not converge to a poorly-connected configuration, we extend  $\Phi_{R}$ by adding a second rule for randomizing the links between the agents. The purpose of this additional rule is to induce a uniform limiting distribution over the set of all connected $k$-regular graphs instead of having the network converge to an element of that set. As such, the resulting graph will be a random $k$-regular graph. A locally applicable method was introduced in \cite{Mahlmann05} for transforming regular graphs into random regular graphs. We define an additional rule for $\Phi_R$ based on this method. The resulting grammar, $\Phi_{RR}$, uses the same labels as $\Phi_R$, and it is defined as 

$$
\includegraphics[scale=0.39]{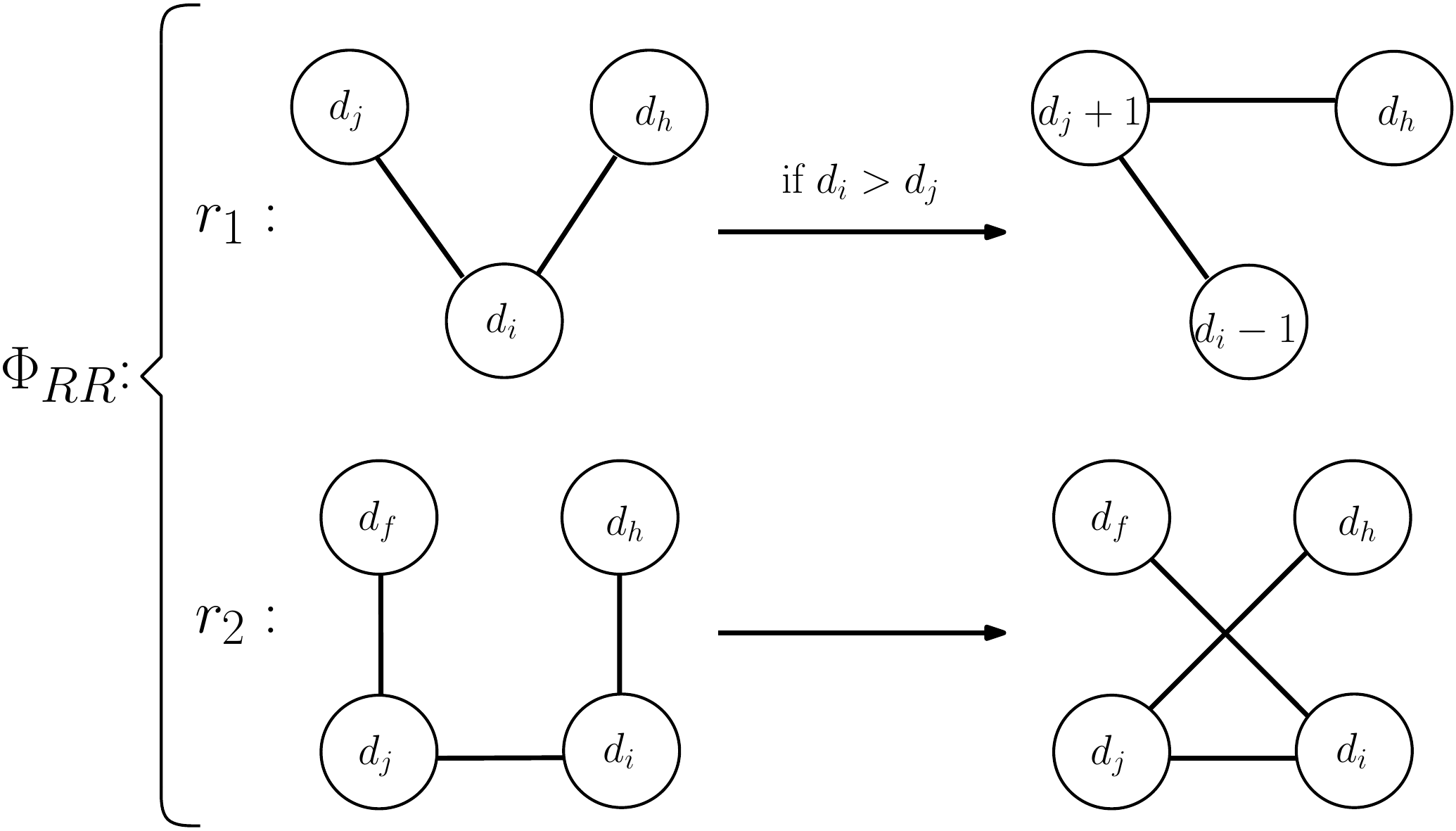}
$$
where $d_i$, $d_j$, $d_h$, and $d_f$ denote the degrees of the corresponding nodes. In accordance with $\Phi_{RR}$, if a node has more links than one of its neighbors has, then it rewires one of its other neighbors to its less-connected neighbor ($r_1$). Also, adjacent nodes exchange their exclusive neighbors ($r_2$). 

%

Both rules in $\Phi_{RR}$ maintain the number of edges and the connectivity.  Unlike $\Phi_R$, a connected graph is never stationary under $\Phi_{RR}$ (unless it is a complete graph). Once a regular graph is reached, $r_1$ is not applicable anymore, but the agents keep randomizing their links via $r_2$. 


There may be many feasible concurrent applications of $\Phi_{RR}$ on an interaction graph. In such cases, the agents need to execute one of the alternatives randomly. These randomizations do not have to be uniform, and assigning each alternative a non-zero probability is sufficient for the desired limiting behavior to emerge. We provide Algorithm \ref{RandRegAlg} as a such distributed implementation that leads to a uniform limiting distribution over the set of connected $k$-regular graphs. 

Algorithm \ref{RandRegAlg} is memoryless since each iteration only depends on the current graph, and the probability of any feasible transition is independent of the past. As such, it induces a Markov chain, $P_{RR}$, over the set of connected graphs with the same average degree as the initial graph, i.e.

\begin{equation}
\label{statespace}
\mathbb{G}_{n,k}=\{ \mathcal{G}=(V,E) \mid |V| = n, |E|= \frac{kn}{2} \},
\end{equation}
where $k= \bar{d}(\mathcal{G}(0))$. Let $\mathbb{G}^0_{n,k} \subseteq \mathbb{G}_{n,k}$ be the set of connected $k$-regular graphs. Note that if $k \in \mathbb{N}$, then $\mathbb{G}^0_{n,k} \neq \emptyset $. A thorough analysis of the dynamics induced by Algorithm \ref{RandRegAlg} along with the following result can be found in \cite{Yasin14CDC}. 

\begin{theorem}
\label{limdist} Let $\mathbb{G}_{n,k}$ satisfy $k \in \mathbb{N}$. Then, $P_{RR}$ has a unique limiting distribution, $\mu^*$, satisfying 
\begin{equation}
\label{limdisteq}
\mu^*(\mathcal{G})= \left\{\begin{array}{ll}1/|\mathbb{G}^0_{n,k}|&\mbox{ if $\mathcal{G} \in \mathbb{G}^0_{n,k}$}, 
 \\\;\;\; 0&\mbox{ otherwise. }\end{array}\right.
\end{equation}
\end{theorem}

\begin{spacing}{1.5}
\begin{center}
\resizebox{6.7cm}{!}{
\begin{alg}{|l|}
\hline
\label{RandRegAlg}
\textbf{Algorithm \Roman{algno}}: Distributed Implementation of $\Phi_{RR}$ \\
\hline
\mbox{\small $\;1:\;$}\textbf{initialize:} $\mathcal{G}=(V,E)$ connected, $\epsilon \in (0,1)$ \\
\mbox{\small $\;2:\;$}\textbf{repeat} \\
\mbox{\small $\;3:\;$}\hspace{0.2cm} Each agent, $i$, is active with probability $1-\epsilon$. \\
\mbox{\small $\;4:\;$}\hspace{0.2cm} Each active $i$ picks a random $j \in \mathcal{N}_i$. \\
\mbox{\small $\;5:\;$}\hspace{0.2cm} For each $i$, $R_i = \{i' \in \mathcal{N}_i \mid i' \text{ picked } i\}$. \\
\mbox{\small $\;6:\;$}\hspace{0.2cm} $\textbf{for}\hspace{0.1cm}  $(each $(i,j)$ s.t. $i \in R_j, j \in R_i$, $d_i\geq d_j$)\\
\mbox{\small $\;7:\;$}\hspace{0.4cm}$\max\{i,j\}$ picks a random $r \in \Phi_{RR}$.\\
\mbox{\small $\;8:\;$}\hspace{0.4cm} $\textbf{if}\hspace{0.1cm}  $($r=r_1$, $d_i>d_j$, $|R_i| \geq 2$)\\
\mbox{\small $\;9:\;$}\hspace{0.6cm} $i$ picks a random $h\in R_i \setminus \{j\}$.\\
\mbox{\small $10:\;$}\hspace{0.6cm} $\textbf{if} \hspace{0.1cm}  $($(j,h) \notin E$)\\
\mbox{\small $11:\;$}\hspace{0.8cm} $E = (E \setminus \{(i,h)\}) \cup \{(j,h)\}$.\\
\mbox{\small $12:\;$}\hspace{0.6cm} \textbf{end if}\\
\mbox{\small $13:\;$}\hspace{0.4cm} $\textbf{else if} \hspace{0.1cm}  $($r=r_2$, $|R_i|\geq 2$, $|R_j| \geq 2$)\\
\mbox{\small $14:\;$}\hspace{0.6cm} $i$ picks a random $h\in R_i \setminus \{j\}$.\\
\mbox{\small $15:\;$}\hspace{0.6cm} $j$ picks a random $f\in R_j\setminus \{i\}$.\\
\mbox{\small $16:\;$}\hspace{0.6cm} $\textbf{if} \hspace{0.1cm}  $($(i,f) \notin E$, $(j,h) \notin E$)\\
\mbox{\small $17:\;$}\hspace{0.8cm} $E = (E \setminus \{(i,h), (j,f)\}) \cup \{(i,f),(j,h)\}$.\\
\mbox{\small $18:\;$}\hspace{0.6cm} \textbf{end if}\\


\mbox{\small $19:\;$}\hspace{0.4cm} \textbf{end if}\\
\mbox{\small $20:\;$}\hspace{0.2cm} \textbf{end for}\\
\mbox{\small $21:\;$}\textbf{end repeat}\\
\hline
\end{alg}
}
\end{center}
\end{spacing}


In light of Theorem \ref{limdist}, any connected initial graph with an integer average degree, $k \in \mathbb{N}$, can be transformed into a connected random $k$-regular graph via $\Phi_{RR}$. For any $k \geq 3$, such graphs have their algebraic connectivity and expansion ratios bounded away from zero, even when the network size is arbitrarily large. If $k \notin \mathbb{N}$, then the degree range converges to 1 via $r_1$, as given in Theorem \ref{degreg}. In that case, the limiting graph can be loosely described as a random almost-regular graph.  While such graphs have many structural similarities to random regular graphs, there is not an exact characterization of their expansion properties. In the next part, we will build on $\Phi_{RR}$ to obtain a grammar that produces random regular graphs, even when $k \notin \mathbb{N}$.

\subsection{Average Degree Manipulation}


Both $\Phi_R$ and $\Phi_{RR}$ were designed to maintain the overall network sparsity. In order to obtain random regular graphs even when the initial graph has a non-integer average degree,  $k \notin \mathbb{N}$, the agents need to manipulate the total number of edges. To this end, the sparsity constraint should be relaxed to allow for occasional increases and decreases in the number of edges as long as the graph is not regular and the average degree is within some proximity of its initial value. Also, the feasible range of the number of edges should ensure that the average degree can reach an integer value. Since an $m$-regular graph with an odd number of nodes does not exist if $m$ is odd, an even integer should always exist within the feasible range of the average degree. In this part, we derive such a grammar, $\Phi^*$, by building on $\Phi_{RR}$. 





 
 In  $\Phi^*$, each node, $i\in V$, has a label, $l(i)$, consisting of two elements. The first entry of the label is equal to the degree of the node, just as in $\Phi_R$ and $\Phi_{RR}$. The second entry of $l(i)$ is a binary flag that denotes whether the corresponding agent is allowed to locally increase or decrease the number of edges in the network. A node, $i$, can add an extra edge to the graph only if $w_i=0$, and it can remove an edge only if $w_i=1$. As such, the node labels are defined as
\begin{equation}
\label{label_dw}
l(i)=\left[\begin{array}{l} d_i \\ w_i \end{array}\right], \forall i\in V,
\end{equation}
where  $d_i$ denotes the degree of $i$, and $w_i \in \{0,1\}$. For any graph, the degree of each node is already encoded through the edge set, $E$. Hence, for the sake of simplicity, a labeled graph will be denoted as $\mathcal{G}=(V,E,w)$ in the sequel.
Using the node labels given in (\ref{label_dw}), the proposed grammar, $\Phi^*= \{r_1, r_2, r_3, r_4\}$, is defined as 

$$
\includegraphics[scale=0.37]{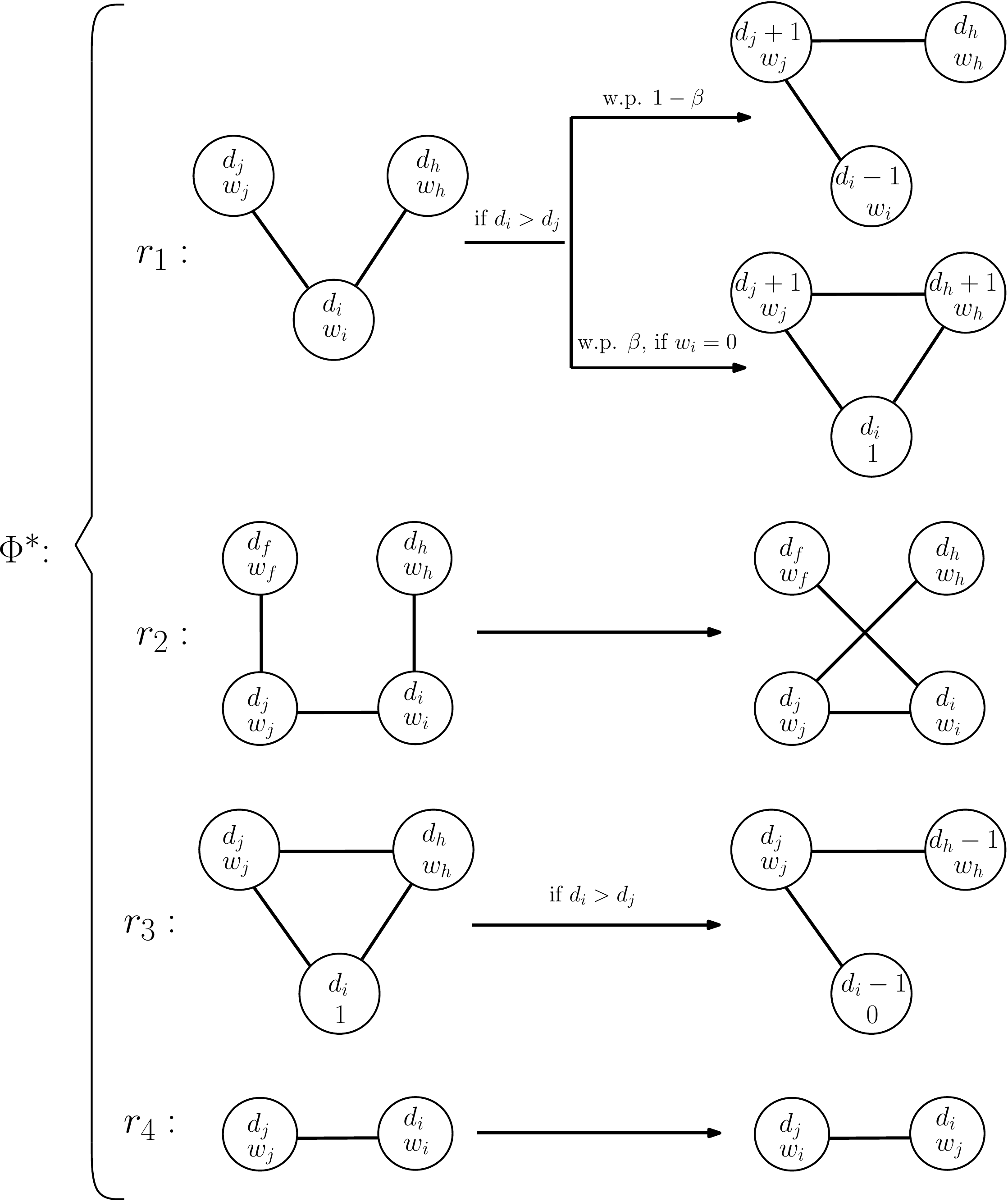}
$$

The rules in $\Phi^*$ can be interpreted as follows: The first rule, $r_1$, is a variant of the degree balancing rule in $\Phi_{RR}$. While $r_1$ induces a transformation identical to the degree balancing rule in $\Phi_{RR}$ with a probability, $1-\beta$, there is also a nonzero probability, $\beta$, that if the higher degree node has its binary flag equal to 0, then it maintains its link with the neighbor introduced to its lower degree neighbor. Note that the total number of edges increases in that case, and the binary flag of the corresponding higher degree node becomes 1. The second rule, $r_2$, is the same link randomization rule as in $\Phi_{RR}$. The third rule, $r_3$, is the edge removal rule which breaks one edge of a triangle if all of its nodes do not have the same degree and one of the higher degree nodes has its binary flag equal to 1. Note that the total number of edges decreases in that case, and the binary flag of the corresponding higher degree node becomes 0. The final rule, $r_4$, is for exchanging the binary flags among adjacent nodes. The purpose of $r_4$ is to ensure that the edge additions and removals will be applicable as long as they are needed to drive the average degree to an integer. In the remainder of this section, the dynamics induced by $\Phi^*$ is analyzed. 



\begin{lem}
\label{cons_3r} Graph connectivity is maintained under any concurrent application of $\Phi^*$. 
\end{lem}
\begin{proof}
Let a connected graph $\mathcal{G}(t)$ be transformed into $\mathcal{G}(t+1)$ via some concurrent application of $\Phi^*$. Since $\mathcal{G}(t)$ is connected, for every pair of nodes, $v,v' \in V$, $\mathcal{G}(t)$ contains a finite path $P$ between $v$ and $v'$. If all the edges traversed in $P$ are maintained in $\mathcal{G}(t+1)$, then $P$ is also a valid path on $\mathcal{G}(t+1)$. Otherwise, any missing edge is removed due to an execution of $r_1$, $r_2$, or $r_3$.  Accordingly, the following procedure can be executed to transform $P$ into a valid path on $\mathcal{G}(t+1)$ between the same terminal nodes as $P$: 

For each edge that is traversed in $P$ but missing in $\mathcal{G}(t+1)$,

1) If the edge was removed via $r_1$ or $r_3$, then replace the corresponding part of the path, $\{i,h\}$, with $\{i,j,h\}$.

2) If the edge was removed via $r_2$, then replace the corresponding part of the path, $\{i,h\}$ (or $\{j,f\}$), with $\{i,j,h\}$ (or $\{j,i,f\}$).
 
Consequently, if $\mathcal{G}(t)$ is connected, then $\mathcal{G}(t+1)$ is also connected.

\end{proof}

Note that, unlike $\Phi_R$ and $\Phi_{RR}$, the the number of edges is not maintained under $\Phi^*$ due to the possible edge additions and removals. However, the number of edges is bounded in an interval through the binary labels, $w$.  

\begin{lem}
\label{edgew} Let $\mathcal{G}=(V,E,w)$ be a graph, and let $\tau$ be any feasible trajectory of $(\mathcal{G}, \Phi^*)$. For any $\mathcal{G}'=(V,E',w') \in \tau$,
 \begin{equation}
\label{edge_dif}
|E'|-|E| = \bold{1}^Tw' - \bold{1}^Tw,
\end{equation}
where $\bold{1}$ is the vector of all ones.
\end{lem}
\begin{proof}
%
%
%

Let $\mathcal{G}(t)=(V,E(t),w(t))$ and $\mathcal{G}(t+1)=(V,E(t+1),w(t+1))$ be a pair of consecutive graphs on any trajectory, $\tau$, of $(\mathcal{G}, \Phi^*)$. Note that the number of edges is maintained under $r_2$ and $r_4$, and it can increase by 1 only due to $r_1$ or decrease by 1 only due to $r_3$. If an application of $r_1$ adds an extra edge to the network, then one of the nodes involved in that transformation, $i$, satisfies $w_i(t+1)-w_i(t)=1$, whereas the other two participating nodes maintain their $w$ entries. On the other hand, for each application of $r_3$, one of the nodes involved in that transformation, $i$, satisfies $w_i(t+1)-w_i(t)=-1$ while the other two participating nodes maintain their $w$ entries. Furthermore, any rule application that maintains the number of edges also maintains the sum of $w$ entries of the participating nodes. In particular, any edge rewiring via $r_1$ or any neighbor swapping via $r_2$ maintain the $w$ entries of the participating nodes, whereas any label exchange via  $r_4$ maintains the sum of labels. Hence, $E(t)$ and $E(t+1)$ satisfy
 \begin{equation}
\label{edge_dif2}
|E(t+1)|-|E(t)| = \bold{1}^Tw(t+1) - \bold{1}^Tw(t).
\end{equation}
Using induction, it can be shown that (\ref{edge_dif2}) implies (\ref{edge_dif}) for any $\mathcal{G}'=(V,E',w') \in \tau$.
  \end{proof}

\begin{cor}
\label{edgewcor}Let $\mathcal{G}=(V,E,\bold{0})$ be a graph, and let $\tau$ be any feasible trajectory of $(\mathcal{G}, \Phi^*)$. For any $\mathcal{G}'=(V,E',w') \in \tau$,
 \begin{equation}
\label{edge_b}
|E| \leq |E'| \leq |E|+|V|.
\end{equation}
\end{cor}
\begin{proof}
Let $\mathcal{G}=(V,E,\bold{0})$ be a graph, and let $\tau$ be any feasible trajectory of $(\mathcal{G}, \Phi^*)$.  In light of Lemma \ref{edgew}, every $\mathcal{G}'=(V,E',w') \in \tau$ satisfies
\begin{equation}
\label{edge_b2}
|E'|-|E| =   \bold{1}^Tw' - \bold{1}^T\bold{0} = \bold{1}^Tw'.
\end{equation}
Since each entry of $w$ is either equal to 1 or equal 0,
\begin{equation}
\label{edge_b3}
0 \leq \bold{1}^Tw' \leq |V|.
\end{equation}
Hence, (\ref{edge_b2}) and (\ref{edge_b3}) together imply (\ref{edge_b}).
\end{proof}


Corollary \ref{edgewcor} provides the upper and the lower bounds on the of number of edges along any trajectory of $(\mathcal{G},\Phi^*)$ for any $\mathcal{G}=(V,E,\bold{0})$. Lemma \ref{intdbar} shows that the corresponding interval always contains a value that implies an integer average degree. 



\begin{lem}
\label{intdbar} For any $\mathcal{G}=(V,E)$, a graph, $\mathcal{G}'=(V,E')$, satisfying $\bar{d}(\mathcal{G}') \in \mathbb{N}$ can be formed by adding at most $|V|$ edges to $\mathcal{G}$.
\end{lem}
\begin{proof}
Let $\mathcal{G}=(V,E)$ be a graph, and let $\mathcal{G}'=(V,E')$ be a graph that is formed by adding $\eta$ edges to $\mathcal{G}$. Then,
\begin{equation}
\label{intbareq2}
\bar{d}(\mathcal{G}')=\frac{2(|E|+\eta)}{|V|}. 
\end{equation}
 For any $\mathcal{G}=(V,E)$, there exists unique pair of integers, $p ,q \in \mathbb{N}$, such that $0 \leq q < |V|$ and
\begin{equation}
\label{intbareq3}
|E| = p |V| + q.
\end{equation}
Let $\eta^* = |V|-q$. Note that $\eta^*\leq |V|$ since $0 \leq q < |V|$. Furthermore, by plugging $\eta^*$ into (\ref{intbareq2}), we obtain
\begin{equation}
\label{intbareq4}
\frac{2(|E|+\eta^*)}{|V|} = 2 p + 2\frac{q}{|V|} +2 -2\frac{q}{|V|} = 2(p+1).
\end{equation}
Since $p\in \mathbb{N}$, we have $2(p+1) \in \mathbb{N}$. Hence, for any $\mathcal{G}=(V,E)$, there exists $0 \leq \eta \leq \eta^*$ such that a graph, $\mathcal{G}'=(V,E')$, satisfying $\bar{d}(\mathcal{G}') \in \mathbb{N}$ can be obtained by adding $\eta$ edges to $\mathcal{G}$.
\end{proof}

In light of Lemma \ref{intdbar}, the interval in (\ref{edge_b}) always has at least one value implying an integer average degree.  Next, it is shown that $(\mathcal{G}, \Phi^*)$ almost surely reaches a regular graph for any connected $\mathcal{G}=(V,E,\bold{0})$ such that $\bar{d}(\mathcal{G})>2$. 


\begin{lem}
\label{increase_edges} Let $\mathcal{G}=(V,E,w)$ be a connected graph such that $\bar{d}(\mathcal{G}) \notin \mathbb{N}$. If $\bold{1}^Tw < |V|$, then $\mathcal{G}$ can be transformed via $\Phi^*$ into a graph, $\mathcal{G}'=(V,E',w')$, satisfying $|E'|=|E|+1$.
\end{lem}
\begin{proof} If $\bar{d}(\mathcal{G}) \notin \mathbb{N}$, then $\mathcal{G}$ has to be a non-regular graph. Furthermore, since $\mathcal{G}$ is connected, then there has to be at least a pair of neighboring nodes, $i$ and $j$, with different degrees. Without loss of generality let $d_i>d_j$. Then, $i$ should have at least one neighbor that is not linked to $j$. As such, $r_1$ is applicable on $\mathcal{G}=(V,E)$. Furthermore,  since each entry of $w$ is either 0 or 1,  at least one entry of $w$ is equal to 0 if $\bold{1}^Tw<|V|$.  If $w_i=0$, then an extra edge can be added to the graph via $r_1$ with probability $\beta$. If $w_i \neq 0$, then there exists another node, $i'$, such that $w_{i'}=0$. Apply $r_4$ along the shortest path between $i$ and $i'$ to obtain $w_i=0$. Once $w_i=0$ is obtained, then with probability $\beta$ an extra edge can be added to the graph in accordance with $r_1$ to obtain a graph, $\mathcal{G}'=(V,E',w')$, satisfying $|E'|=|E|+1$.
\end{proof}

\begin{lem}
\label{triangle_form} Let $\mathcal{G}=(V,E,w)$ be a connected triangle-free graph. If $\bar{d}(\mathcal{G})>2$, then $\mathcal{G}$ can be transformed via $\Phi^*$ into a graph, $\mathcal{G}'=(V,E',w)$, containing at least one triangle.
\end{lem}
\begin{proof} Let $\mathcal{G}=(V,E,w)$ be a connected triangle-free graph. Since $\bar{d}(\mathcal{G})>2$, $\mathcal{G}$ has to be a cyclic graph. Consider the shortest cycle on  $\mathcal{G}$, $C_s$. Note that $C_s$ has to contain at least 4 nodes, since $\mathcal{G}$ is triangle-free. Furthermore, if $\bar{d}(\mathcal{G})>2$, then $C_s$ cannot contain all the nodes of  $\mathcal{G}$, as otherwise $\mathcal{G}$ has to be a cycle graph with $\bar{d}(\mathcal{G})=2$. Also, since $\mathcal{G}$ is a connected graph, at least one node on $C_s$ must be linked to a node that is not contained in $C_s$.  Without loss of generality, let $i$ be a node on $C_s$ connected to an off-cycle node, $h$. Furthermore, let $j$ denote a neighbor of $i$ on $C_s$. Note that $h$ can not be connected to $j$ since $\mathcal{G}$ is triangle-free.  Furthermore, since the $C_s$ contains at least 4 nodes, $j$ has to have another neighbor on $C_s$, $f$, which is not connected to $i$. If $i$ and $j$ execute $r_2$ by rewiring $h$ and $f$ to each other, $C_s$ becomes a shorter cycle. Note that $w$ is invariant under $r_2$. Hence, by repeating this process until $C_s$ consists of 3 nodes, $\mathcal{G}$ can be transformed into a graph, $\mathcal{G}'=(V,E',w)$, containing at least one triangle. 
\end{proof}

\begin{lem}
\label{decrease_edges} Let $\mathcal{G}=(V,E,w)$ be a connected graph such that $\bar{d}(\mathcal{G}) \notin \mathbb{N}$. If $\bold{1}^Tw > 0$ and $\bar{d}(\mathcal{G})>3$, then $\mathcal{G}$ can be transformed via $\Phi^*$ into a graph, $\mathcal{G}'=(V,E',w')$, satisfying $|E'|=|E|-1$.
\end{lem}
\begin{proof} Let $\mathcal{G}=(V,E,w)$ be a connected graph such that $\bar{d}(\mathcal{G}) \notin \mathbb{N}$. Since $\bar{d}(\mathcal{G}) \notin \mathbb{N}$, the degree range can be reduced to 1 without changing $\bar{d}$ and $w$ via $r_1$, as given in Theorem \ref{degreg}. Furthermore, in light of Lemma \ref{triangle_form}, if the resulting graph is triangle-free, then it can be transformed via $r_2$ into a graph, $\mathcal{G}^{+}=(V,E^{+},w)$, which contains at least one triangle. Since  $\Delta(\mathcal{G}^+)-\delta(\mathcal{G}^+)=1$ and $\bar{d}(\mathcal{G}^+)=\bar{d}(\mathcal{G})>3$, we have $\delta(\mathcal{G}^+)\geq3$.

If $\bold{1}^Tw> 0$, then at least one entry of $w$ is equal to 1. As such, $w_i=1$ can be reached for any $i \in V$ via a sequence of $r_4$ applications. Hence, if any of the triangles in $\mathcal{G}^{+}$ involves nodes with non-uniform degrees, then one edge of the triangle can be broken in accordance with $r_3$ to obtain a graph, $\mathcal{G}'=(V,E')$, satisfying $|E'|=|E|-1$. 

On the other hand, assume that all the triangles on $\mathcal{G}^+$ consist of nodes with uniform degrees. Let $i$, $j$, and $h$ be the nodes of such a triangle. Consider the shortest path from any node of this triangle to a node whose degree is not equal to the degrees of the triangle nodes. Without loss of generality, let this path be between $i$ and $i'$ such that $d^+_{i'} \neq d^+_i$. Note that each node other than $i'$ on this path has its degree equal to $d^+_i$, and either $d^+_i>d^+_{i'}$ or $d^+_i<d^+_{i'}$. Both cases are inspected below: 

1) If $d^+_i>d^+_{i'}$, then $d^+_i \geq 4$ given that $\delta(\mathcal{G}^+)\geq3$. Hence, the node next to $i'$ on the shortest path has at least one neighbor outside the shortest path that is not connected to $i'$. Applying $r_1$ by rewiring such a neighbor to $i'$ results in a shorter path from $i$ to a smaller degree node. Applying the same procedure eventually results in a smaller degree node being adjacent to $i$. Since $d^+_i \geq 4$, $i$ has at least one neighbor other than $j$ and $h$ to rewire to its smaller degree neighbor. After the corresponding application of $r_1$, we obtain a triangle of nodes with non-uniform degrees.  Hence, one edge of the triangle can be broken in accordance with $r_3$ to obtain a graph, $\mathcal{G}'=(V,E',w')$, satisfying $|E'|=|E|-1$. 

2) If $d^+_i<d^+_{i'}$, then $i'$ can apply $r_1$ to rewire one of its neighbors to the node on the shortest path next to itself. This results in a shorter path from $i$ to higher degree node is obtained. Applying the same procedure eventually increases the degree of $i$, and we obtain a triangle of nonuniform degrees. Hence, one edge of the triangle can be broken in accordance with $r_3$ to obtain a graph, $\mathcal{G}'=(V,E',w')$, satisfying $|E'|=|E|-1$.
\end{proof}

\begin{lem}
\label{reg_feas}  Let $\mathcal{G}=(V,E,\bold{0})$ be a connected graph. If $\bar{d}(\mathcal{G})>2$, then $(\mathcal{G},\Phi^*)$ almost surely reaches an $m$-regular graph such that $\bar{d}(\mathcal{G})\leq m \leq \bar{d}(\mathcal{G})+2$.
\end{lem}
\begin{proof} Let $\mathcal{G}=(V,E,\bold{0})$ be a connected graph such that $\bar{d}(\mathcal{G})>2$. Due to Corollary \ref{edgewcor}, the number of edges stays in $\{|E|, |E|+1, \hdots |E|+|V| \}$ along any trajectory of $(\mathcal{G},\Phi^*)$. In light of Lemma \ref{intdbar}, this interval contains at least one value, $|E|+\eta$, such that the corresponding average degree, $m$, is an integer satisfying

\begin{equation}
\label{meq}
\bar{d}(\mathcal{G}) \leq m = \frac{2(|E|+\eta)}{|V|}\leq \bar{d}(\mathcal{G})+2.
\end{equation}
 As such, $m \geq 3$ since $\bar{d}(\mathcal{G})>2$. Let $\mathcal{G}'=(V,E',w')$ be a graph reached from $\mathcal{G}$ via $\Phi^*$ such that $\bar{d}(\mathcal{G}') \notin \mathbb{N}$. Then, either $|E'|<|E|+\eta$ or $|E'|>|E|+\eta$. Both cases are inspected below: 

1) If $|E'|<|E|+\eta$, then $\bold{1}^Tw'=|E'|-|E| < |V|$. As such, in light of Lemma \ref{increase_edges}, $\mathcal{G}'$ can have its number of edges increased by 1 via $\Phi^*$. This process can be repeated unless the average degree reaches an integer value. 

2) If $|E'|>|E|+\eta$, then $\bold{1}^Tw'=|E'|-|E| > 0$. Furthermore, $\bar{d}(\mathcal{G}')>3$ since $m\geq 3$. As such, in light of Lemma \ref{decrease_edges}, $\mathcal{G}'$ can have its number of edges decreased by 1 via $\Phi^*$. This process can be repeated unless the average degree reaches an integer value.

Hence, in either case it is possible to add or remove edges to the network via $\Phi^*$ until the number of edges implies an integer average degree, $m \in[ \bar{d}(\mathcal{G}), \bar{d}(\mathcal{G})+2]$. Once an average degree of $m$ is obtained, then the graph can be driven to a $m$-regular configuration through a sequence of $r_1$ applications. As such, any $\mathcal{G}'$ has a non-zero probability of reaching an $m$-regular graph after a finite sequence of $\Phi^*$ applications. Consequently, $(\mathcal{G},\Phi^*)$ almost surely reaches an $m$-regular graph such that $\bar{d}(\mathcal{G})\leq m \leq \bar{d}(\mathcal{G})+2$.
\end{proof}

Unlike $\Phi_R$ and $\Phi_{RR}$, the degree range does not monotonically decrease under $\Phi^*$ due to the possible edge addition and removals. For instance, increasing the total number of edges via $r_1$ would increases the degree range if $d_h= \Delta(\mathcal{G})$ and $\delta(\mathcal{G})<d_j$. However, if a regular graph is reached, then the graph remains regular since neither $r_1$ nor $r_3$ are applicable on a regular graph.

\begin{lem}
\label{regstaysreg} Let $\mathcal{G}$ be an $m$-regular graph. Any feasible trajectory of $(\mathcal{G}, \Phi^*)$ consists of only $m$-regular graphs. 
\end{lem}
\begin{proof}
On an $m$-regular graph,  $r_1$ and $r_3$ are not applicable since $d_i=d_j$ for any pair of nodes, $i,j$. Furthermore, since the node degrees are invariant to the applications of $r_2$ and $r_4$, any feasible trajectory of $(\mathcal{G}, \Phi^*)$ consists of only $m$-regular graphs.
\end{proof}
 
Once an $m$-regular graph is reached, the graph evolves only under $r_2$ and $r_4$. As such, it keeps randomizing within the corresponding set of $m$-regular graphs.  The induced limiting distribution depends on how the agents implement $\Phi^*$. In the next section, we present a distributed implementation that leads to a uniform limiting distribution over the corresponding set of $m$-regular graphs.
\section{Distributed Implementation} 
\label{distalg}
 In this section, we present Algorithm \ref{RoundRandRegAlg} as a distributed implementation of $\Phi^*$ and analyze the resulting dynamics. 

In accordance with Algorithm \ref{RoundRandRegAlg}, the nodes behave as follows: At each iteration, each node is inactive with a probability, $\epsilon \in (0,1)$. The inactivation probability, $\epsilon$, ensures that any feasible application of $\Phi^*$ can be realized through Algorithm \ref{RoundRandRegAlg}, as it will be shown in Lemma \ref{phicouples3R}. Inactive nodes do not participate in any rule execution in that time step.  Each active agent, $i$, picks one of its neighbors, $j \in \mathcal{N}_i$, uniformly at random, and it communicates its degree to that neighbor. Through these communications, each active agent, $i$, obtains the list of neighbors that picked itself, $R_i \subseteq \mathcal{N}_i$, and checks if it is matched, i.e. if $j\in R_i$. If that is not the case, then $i$ is a follower in that time step, i.e. it will not initiate a rule execution but it will participate if $j$ wants to rewire $i$ to some other node. If $j\in R_i$, then $i$ and $j$ are matched. Each matched pair randomly pick a candidate rule, $r \in \Phi^*$, that they will potentially execute. In Algorithm \ref{RoundRandRegAlg}, without loss of generality, the candidate rule is picked by the agent with the larger node ID, i.e. $\max\{i,j\}$. Depending on the chosen rule, $r \in \Phi^*$, the matched pair of nodes, $i$ and $j$, do one of the following:
\begin{enumerate}
\item $r=r_1$:  If $d_i>d_j$ and $|R_i|\geq 2$, then $i$ chooses a neighbor, $h \neq j \in R_i$, uniformly at random. If $h \notin \mathcal{N}_j$, then, with probability $1-\beta$, $h$ is rewired to $j$. With the remaining probability $\beta$, if $w_i=0$, a link is formed between $h$ and $j$, the link between $i$ and $h$ is maintained, and $w_i'=1$.

\item $r=r_2$: If $|R_i|, |R_j| \geq 2$, both $i$ and $j$ choose one neighbor, $h \neq j  \in R_i$ and $f  \neq i \in R_j$, uniformly at random. If neither $h$ nor $f$ is linked to both $i$ and $j$, then $r_2$ is executed by rewiring $h$ to $j$ and $f$ to $i$.  

\item $r=r_3$: If $d_i>d_j$, $|R_i|\geq 2$), and  $w_i=1$, then $i$ chooses a neighbor, $h \neq j \in R_i$, uniformly at random. If $h \in \mathcal{N}_j$, then the link between $i$ and $h$ is removed and $w_i'=0$. 

\item $r=r_4$: $i$ and $j$ swap their binary flags, i.e. $w_i'=w_j$ and $w_j'=w_i$. 

\end{enumerate}
A feasible iteration of the algorithm on a network is illustrated in Fig. \ref{sample_flow}, where $\mathcal{G}=(V,E,w)$ is transformed into $\mathcal{G}'=(V, E', w')$ in one time step.

\begin{spacing}{1.5}
\begin{center}
\resizebox{6.7cm}{!}{
\begin{alg}{|l|}
\hline
\label{RoundRandRegAlg}
\textbf{Algorithm \Roman{algno}}: Distributed Implementation of $\Phi^*$ \\
\hline
\mbox{\small $\;1:\;$}\textbf{initialize:} $\mathcal{G}=(V,E,w=\bold{0})$ connected, $\epsilon , \beta \in (0,1)$ \\
\mbox{\small $\;2:\;$}\textbf{repeat} \\
\mbox{\small $\;3:\;$}\hspace{0.2cm} Each agent, $i$, is active with probability $1-\epsilon$. \\
\mbox{\small $\;4:\;$}\hspace{0.2cm} Each active $i$ picks a random $j \in \mathcal{N}_i$. \\
\mbox{\small $\;5:\;$}\hspace{0.2cm} For each $i \in V$, $R_i = \{i' \in \mathcal{N}_i \mid i' \text{ picked } i\}$. \\
\mbox{\small $\;6:\;$}\hspace{0.2cm} $\textbf{for}\hspace{0.1cm}  $(each $(i,j)$ s.t. $i \in R_j, j \in R_i$, $d_i\geq d_j$)\\
\mbox{\small $\;7:\;$}\hspace{0.4cm}$\max\{i,j\}$ picks a random $r \in \Phi^*$.\\
\mbox{\small $\;8:\;$}\hspace{0.4cm} $\textbf{if}\hspace{0.1cm}  $($r=r_1$, $d_i>d_j$, $|R_i| \geq 2$)\\
\mbox{\small $\;9:\;$}\hspace{0.6cm} $i$ picks a random $h\in R_i \setminus \{j\}$.\\
\mbox{\small $10:\;$}\hspace{0.6cm} $\textbf{if} \hspace{0.1cm}  $($(j,h) \notin E$)\\
\mbox{\small $11:\;$}\hspace{0.8cm} $i$ picks a random $\beta' \in [0,1]$.  \\
\mbox{\small $12:\;$}\hspace{0.8cm} $\textbf{if} \hspace{0.1cm}  $($\beta'\geq \beta$)\\
\mbox{\small $13:\;$}\hspace{1cm}  $E =  (E \setminus \{(i,h)\}) \cup \{(j,h)\}$.  \\
\mbox{\small $14:\;$}\hspace{0.8cm} $\textbf{else if} \hspace{0.1cm}  $($w_i=0$)\\
\mbox{\small $15:\;$}\hspace{1cm}  $E =  E \cup \{(j,h)\}$, $w_i=1$.  \\
\mbox{\small $16:\;$}\hspace{0.8cm} \textbf{end if}\\
\mbox{\small $17:\;$}\hspace{0.6cm} \textbf{end if}\\

\mbox{\small $18:\;$}\hspace{0.4cm} $\textbf{else if} \hspace{0.1cm}  $($r=r_2$, $|R_i|\geq 2$,$|R_j| \geq 2$)\\
\mbox{\small $19:\;$}\hspace{0.6cm} $i$ picks a random $h\in R_i \setminus \{j\}$.\\
\mbox{\small $20:\;$}\hspace{0.6cm} $j$ picks a random $f\in R_j\setminus \{i\}$.\\
\mbox{\small $21:\;$}\hspace{0.6cm} $\textbf{if} \hspace{0.1cm}  $($(i,f) \notin E$, $(j,h) \notin E$)\\
\mbox{\small $22:\;$}\hspace{0.8cm} $E = (E \setminus \{(i,h), (j,f)\}) \cup \{(i,f),(j,h)\}$.\\

\mbox{\small $23:\;$}\hspace{0.4cm} $\textbf{else if} \hspace{0.1cm}  $($r=r_3$, $d_i>d_j$, $|R_i|\geq 2$, $w_i = 1$)\\
\mbox{\small $24:\;$}\hspace{0.6cm} $i$ picks a random $h\in R_i \setminus \{j\}$.\\
\mbox{\small $25:\;$}\hspace{0.6cm} $\textbf{if} \hspace{0.1cm}  $($(j,h) \in E$)\\
\mbox{\small $26:\;$}\hspace{0.8cm} $E = E \setminus \{(i,h)\}$, $w_i=0$.\\
\mbox{\small $27:\;$}\hspace{0.6cm} $\textbf{end if} $\\

\mbox{\small $28:\;$}\hspace{0.4cm} $\textbf{else if} \hspace{0.1cm}  $($r=r_4$)\\
\mbox{\small $29:\;$}\hspace{0.6cm} Swap $w_i$ and $w_j$.\\

\mbox{\small $30:\;$}\hspace{0.4cm} \textbf{end if}\\
\mbox{\small $31:\;$}\hspace{0.2cm} \textbf{end for}\\
\mbox{\small $32:\;$}\textbf{end repeat}\\
\hline
\end{alg}
}
\end{center}
\end{spacing}

\begin{figure}[htb]
\begin{center}
\includegraphics[trim = 0mm 0mm 0mm 0mm,scale=0.5]{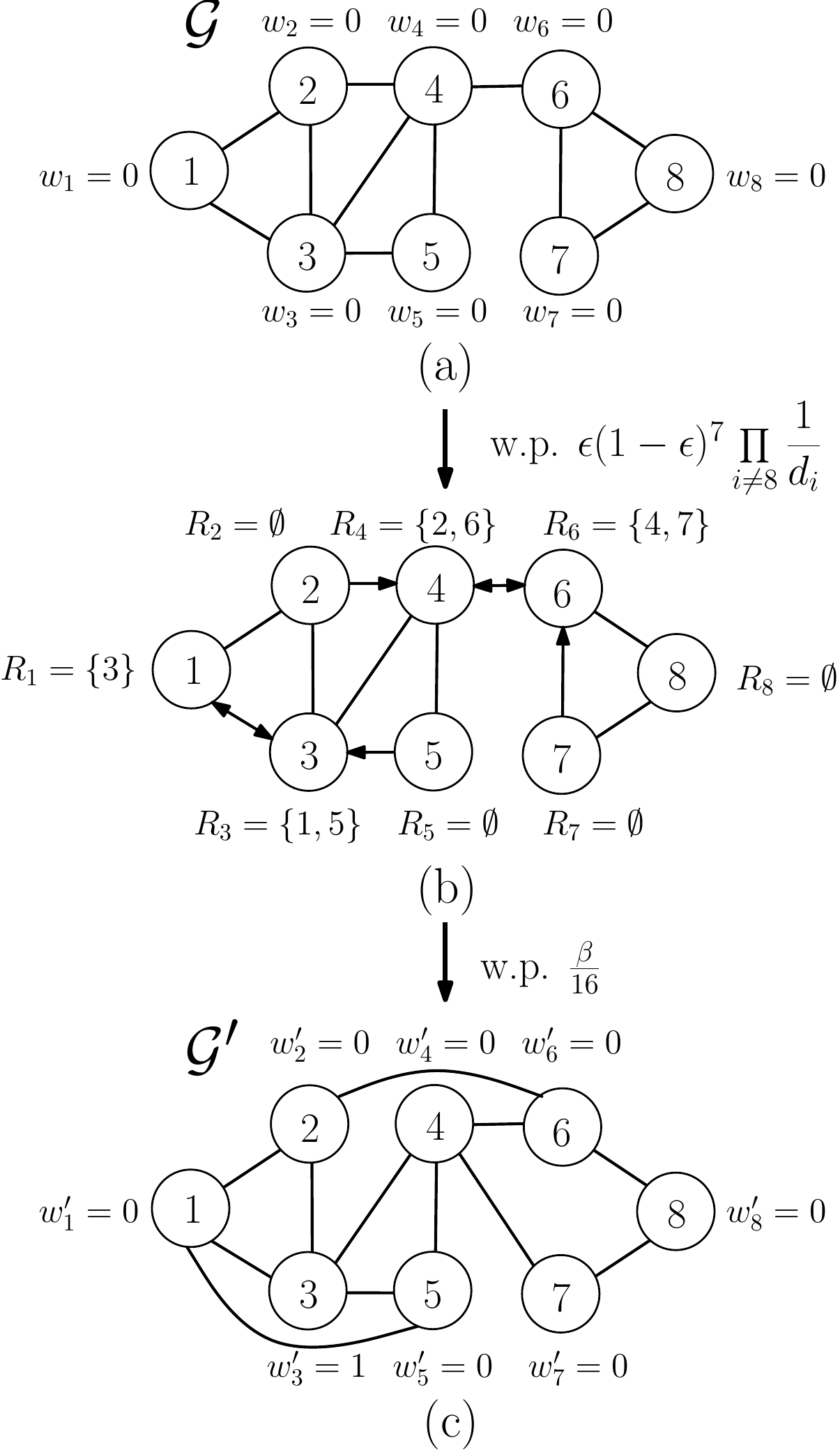}
\caption{Some steps of a feasible iteration of Algorithm \ref{RoundRandRegAlg} on $\mathcal{G}$ in (a) resulting in $\mathcal{G}'$ in (c) along with the probabilities of the corresponding events. In this example, each node other than 8 is active and picks a neighbor as illustrated in (b), where each arrow is pointed from a node to its chosen neighbor. Accordingly, (1,3) and (4,6) are the matched pairs. With a joint probability of $1/16$, node 3 picks $r_1$, and node 6 picks $r_2$ as the candidate rules for their respective matchings. Furthermore, since $R_3\setminus\{1\}=\{5\}$, $R_4\setminus\{6\}=\{2\}$,  and $R_6\setminus\{4\}=\{7\}$; node 3 picks node 5; node 4 picks node 2; and node 6 picks node 7 to rewire. Finally, the edge between node 3 and 5 is maintained with a probability $\beta$ ($w_3=0$). Hence, given (b), $\mathcal{G}'$ can emerge with a probability of $\beta/16$.}
\label{sample_flow}
\end{center}
\end{figure}

Note that Algorithm \ref{RoundRandRegAlg} maintains  connectivity due to Lemma \ref{cons_3r}. Let $k=\bar{d}(\mathcal{G}(0))$ be the average degree of the initial graph, $\mathcal{G}(0)=(V,E(0), \bold{0})$. Then, the average degree remains in $[k,k+2]$ due to Corollary \ref{edgewcor}. Furthermore, Algorithm \ref{RoundRandRegAlg} is memoryless since each iteration only depends on the current graph, and the probability of any feasible transition is independent of the past. Hence, it induces a Markov chain over the state space, $\mathbb{G}_{n,[k,k+2]}$, consisting of connected labeled graphs with $n$ nodes and average degrees contained in $[k, k+2]$, i.e.
\begin{equation}
\label{statespace3R}
\mathbb{G}_{n,[k,k+2]}=\{ \mathcal{G}=(V,E,w) \mid |V| = n, |E|= 0.5kn+\bold{1}^Tw\},
\end{equation}
 where $w_i \in \{0,1\}$ for every $i \in V$. 

 Let $\mu(t)$ denote the probability distribution over the state space at time $t$. Then, $\mu(t)$ satisfies
\begin{equation}
\label{markoveq}
\mu^T(t+1) = \mu^T(t)P^*,
\end{equation}
where $P^*$ is the corresponding probability transition matrix.  Accordingly, the probability of transition from any $\mathcal{G}$ to any $\mathcal{G}'$ is denoted by $P^*(\mathcal{G},\mathcal{G}')$.

By following Algorithm \ref{RoundRandRegAlg}, the agents concurrently modify their neighborhoods in accordance with $\Phi^*$ such that any feasible transformation occurs with a non-zero probability.  In other words, for any pair of graphs, $\mathcal{G}, \mathcal{G}' \in \mathbb{G}_{n,[k,k+2]}$, the corresponding transition probability, $P^*(\mathcal{G},\mathcal{G}')$, is non-zero if it is possible to transform $\mathcal{G}$ into $\mathcal{G}'$ in a single time step via  $\Phi^*$. 


\begin{lem}
\label{phicouples3R}
Let $\mathcal{G}, \mathcal{G}' \in \mathbb{G}_{n,[k,k+2]} $ be any pair of graphs. If $\mathcal{G}'$ can be reached from $\mathcal{G}$ in one step via $\Phi^*$, then $P^*(\mathcal{G},\mathcal{G}')>0$ .
\end{lem}
\begin{proof}
%
%

Let $\mathcal{G}'$ be reachable from $\mathcal{G}$ in one step via $\Phi^*$, and let $G =\{g_1, g_2, \hdots \}$ denote the corresponding set of disjoint subgraphs of $\mathcal{G}$ to be transformed to reach  $\mathcal{G}'$.  Note that for each $g \in G$ there is a $r=(g_l, g_r) \in \Phi^*$ satisfying $g \simeq g_l$. In order to show that $P^*(\mathcal{G},\mathcal{G}')>0$, we present a feasible flow of Algorithm \ref{RoundRandRegAlg} that transforms $\mathcal{G}$ by applying the corresponding $r \in \Phi^*$ to each $g \in G$. For each $g \in G$, let each node in $g$ be active and pick a neighbor as illustrated in Fig. \ref{prooffig2}. 
\begin{figure}[h!]
\begin{center}
\includegraphics[trim = 0mm 0mm 0mm 0mm,scale=0.4]{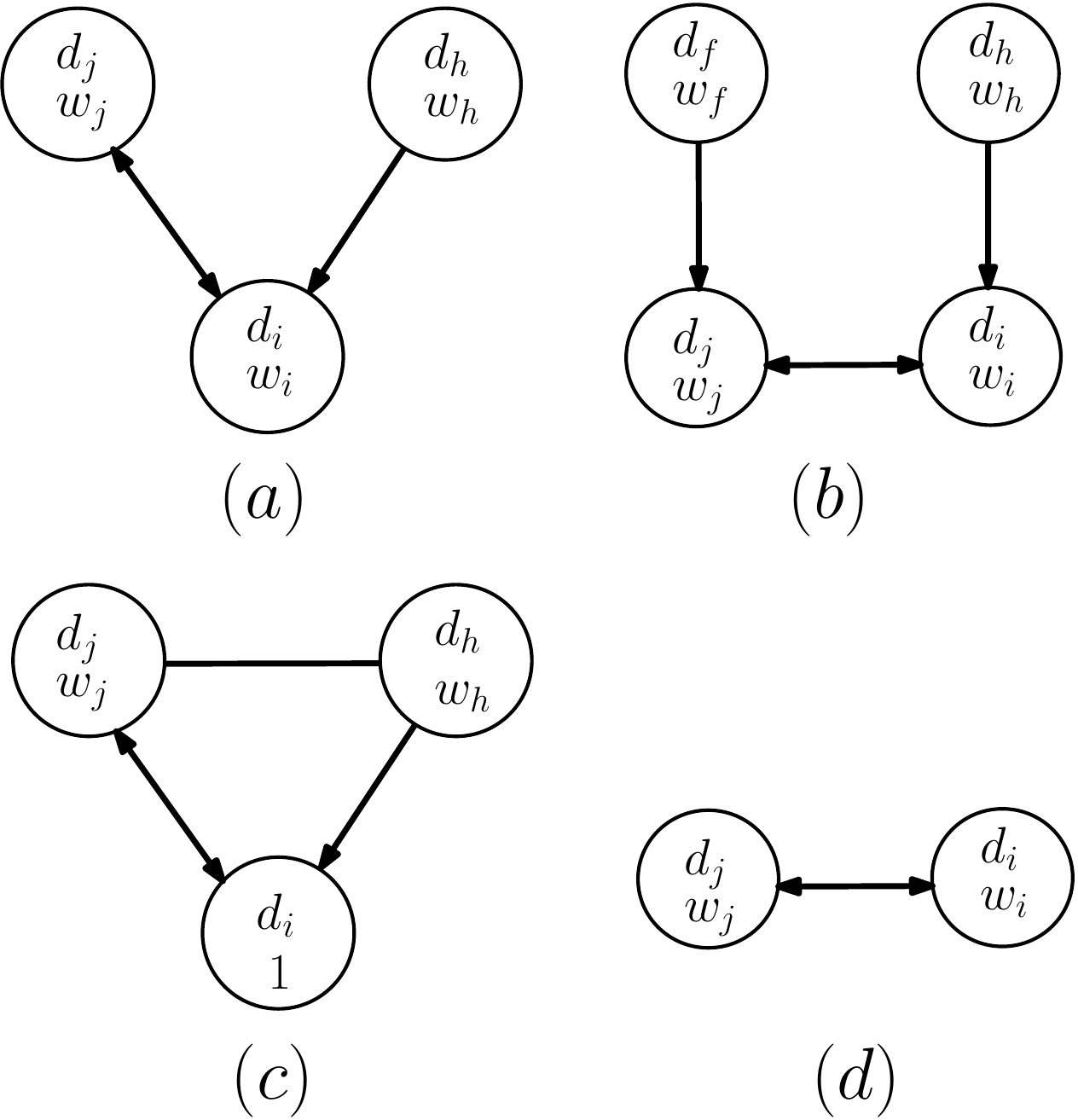}
\caption{An arrow is pointed from each agent to the neighbor it picked. For each $g\in G$, the nodes in $g$ have non-zero probability to pick their neighbors as shown in (a) if $r=r_1$, (b) if $r=r_2$, (c) if $r=r_3$, and (d) if $r=r_4$. }
\label{prooffig2}
\end{center}
\end{figure}

Furthermore, let any node that is not included in any $g \in G$ be inactive, which ensures that only the subgraphs in $G$ will be transformed. Finally, let each $g$ pick the corresponding $r$ as the candidate ruler to execute. In that case, the agents are guaranteed to only apply the corresponding $r \in \Phi^*$ to each $ g \in G$. Hence, the corresponding transformation has a non-zero probability in $P^*$.

 \end{proof}

The state space, $\mathbb{G}_{n,[k,k+2]}$ can be represented as the union of two disjoint sets, $\mathbb{G}^0_{n,[k,k+2]}$ (regular graphs) and  $\mathbb{G}^+_{n,[k,k+2]}$ (non-regular graphs), defined as
  \begin{eqnarray}
\label{statespaceRI3R}
\mathbb{G}^0_{n,[k,k+2]}&=&\{ \mathcal{G}\in \mathbb{G}_{n,[k,k+2]} \mid  f(\mathcal{G}) = 0 \}, \\
\mathbb{G}^+_{n,[k,k+2]} &=& \mathbb{G}_{n,[k,k+2]} \setminus \mathbb{G}^0_{n,[k,k+2]}.
\end{eqnarray}

In light of Lemma \ref{reg_feas}, $\mathbb{G}^0_{n,[k,k+2]} \neq \emptyset$ for $k>2$. Let $M$ denote the set of all such $m$, i.e.
\begin{equation}
\label{mvals}
M= \{m \in \mathbb{N} \mid k \leq m \leq k+2, 0.5m|V| \in \mathbb{N} \}.
\end{equation}

The set of regular graphs, $\mathbb{G}^0_{n,[k,k+2]}$, can be written as the union of some disjoint sets as
\begin{equation}
\label{reg_part}
\mathbb{G}^0_{n,[k,k+2]}= \bigcup_{m \in M} \mathbb{G}^0_{n,m}, 
\end{equation}
where each $\mathbb{G}^0_{n,m}$ is the set of $m$-regular graphs defined as
\begin{equation}
\label{reg_part2}
\mathbb{G}^0_{n,m} = \{\mathcal{G} \in \mathbb{G}^0_{n,[k,k+2]} \mid \bar{d}(\mathcal{G})=m \}.
\end{equation}

In the remainder of this section, the limiting behavior of $P^*$ is analyzed. In light of Lemma \ref{reg_feas}, if $k>2$, then any connected $\mathcal{G}(0)=(V,E(0), \bold{0})$ almost surely converges to  $\mathbb{G}^0_{n,[k,k+2]}$. Next, we show that, for each $m \in M$, $\mathbb{G}^0_{n,m}$ is a closed communicating class.

 \begin{lem}
\label{ComClass3R} If $k>2$, then $\mathbb{G}^0_{n,m}$ is a closed communicating class of $P^*$ for each $m \in M$.  
\end{lem}
\begin{proof}
In Light of Lemma \ref{regstaysreg}, only $r_2$ and $r_4$ is applicable on any $m$-regular graph, and any resulting graph is $m$-regular. Hence, once the system reaches any $\mathbb{G}^0_{n,m}$, it stays in $\mathbb{G}^0_{n,m}$. Furthermore, in light of Theorem \ref{limdist}, any $m$-regular structure can be reached from any other via $r_2$. Moreover, since all the graphs in $\mathbb{G}^0_{n,m}$ are connected, any permutation of the elements in $w$ is also reachable via $r_4$ for each graph structure in $\mathbb{G}^0_{n,m}$. Hence, each $\mathbb{G}^0_{n,m}$ is a closed communicating class of $P^*$.\end{proof} 

The limiting behavior of Algorithm \ref{RoundRandRegAlg} depends on the limiting behavior of $P^*$ in each closed communicating class. Note that for each $\mathbb{G}^0_{n,m}$, there is a unique stationary distribution, $\mu^*_m$, of $P^*$ whose support is $\mathbb{G}^0_{n,m}$. Next, it is shown that each $\mu^*_m$ is a limiting distribution that is uniform over $\mathbb{G}^0_{n,m}$. To this end, first it is shown that the transitions between any two regular graphs are symmetric. 

\begin{lem}
\label{SymReg3R}
For any $m \in M$, let $\mathcal{G}, \mathcal{G}' \in \mathbb{G}^0_{n,m}$ be any pair of $m$-regular graphs. Then,
\begin{equation}
P^*(\mathcal{G}, \mathcal{G}' ) = P^*(\mathcal{G}', \mathcal{G}). 
\end{equation}
\end{lem}
\begin{proof}
For any $m \in M$, let $\mathcal{G}, \mathcal{G}' \in \mathbb{G}^0_{n,m}$ be any pair of $m$-regular graphs. Note that $r_1$ and $r_3$ are not applicable on any graph in $\mathbb{G}^0_{n,m}$ since all the nodes have equal degrees. Hence, any transition from $\mathcal{G}$ to $\mathcal{G}'$ is only via $r_2$ or $r_4$. Since both $r_2$ and $r_4$ are reversible,  $P^*(\mathcal{G}, \mathcal{G}' )>0$ if and only if $P^*(\mathcal{G}', \mathcal{G})>0$.

Let us consider an arbitrary execution of Algorithm \ref{RoundRandRegAlg}, where $\mathcal{G}$ is transformed into $\mathcal{G}'$. Let $u$ be the corresponding vector of randomly picked neighbors in line 4 of Algorithm \ref{RoundRandRegAlg} (let $u_i=null$ if $i$ is inactive). For each node, $i$, let $R_i(u)$ be the set of nodes that picked $i$, and let $M(u)= \{(i,j) \mid i \in R_j(u), j\in R_i(u)\}$ be the set of matched pairs. In the remainder of the proof, it is shown that for any such feasible execution there exists an equally likely execution of Algorithm \ref{RoundRandRegAlg} that transforms $\mathcal{G}'$ back to $\mathcal{G}$.  For $\mathcal{G'}=(V,E',w')$, 
consider the vector, $u'$, whose entries are

\begin{equation}
\label{vprime}
u'_i= \left\{\begin{array}{ll}u_i&\mbox{ if $u_i =null$ or $(i,u_i) \in E'$}, 
 \\ u_{u_i} &\mbox{ otherwise. }\end{array}\right.
\end{equation} 
Note that $\Pr[u]=\Pr[u']$ since the inactive nodes will remain inactive with the same probability, and each active node picks a neighbor uniformly at random. Furthermore, $M(u)=M(u')$. Let each $(i,j)\in M(u)$ randomly choose the same candidate rule, $r$, they executed in the transition from $\mathcal{G}$ to  $\mathcal{G}'$. As such, if $r=r_4$, then $i$ and $j$ will swap $w_i$ and $w_j$ back. On the other hand, if $r=r_2$, then $i$ and $j$ will reverse the neighbor-swapping in the transition from $\mathcal{G}$ to $\mathcal{G}'$ with the same probability (lines 12-17 in Algorithm \ref{RoundRandRegAlg}) since $|R_i(u)| = |R_i(u')|$ and $|R_j(u)| = |R_j(u')|$ for every $(i,j)\in M(u)$. As such, all the local transformations from $\mathcal{G}$ to $\mathcal{G}'$ can be reversed with the same probability under Algorithm \ref{RoundRandRegAlg}. Consequently, $P^*(\mathcal{G}, \mathcal{G}' ) = P^*(\mathcal{G}',\mathcal{G})$.\end{proof}

 \begin{lem}
\label{limincom} For any $m \in M$, there is a unique limiting distribution, $\mu^*_m$, of $P^*$ satisfying 
\begin{equation}
\label{limdist_incom_eq}
\mu^*_m(\mathcal{G})= \left\{\begin{array}{ll}1/|\mathbb{G}^0_{n,m}|&\mbox{ if $\mathcal{G} \in \mathbb{G}^0_{n,m}$}, 
 \\\;\;\; 0&\mbox{ otherwise. }\end{array}\right.
\end{equation}
\end{lem}
\begin{proof}
For any $m \in M$, $\mathbb{G}^0_{n,m}$ is a closed communicating class due to Lemma \ref{ComClass3R}. As such, for each  $\mathbb{G}^0_{n,m}$, there exists a unique stationary distribution, $\mu^*_m$, whose support is $\mathbb{G}^0_{n,m}$. Let $P^{*}_m$ be the $|\mathbb{G}^0_{n,m}|$ by $|\mathbb{G}^0_{n,m}|$ probability transition matrix that only represents the transitions within $\mathbb{G}^0_{n,m}$. Due to Lemma \ref{ComClass3R}, $P^*_m$ is irreducible. Also $P^*_m$ is aperiodic since $P^*(\mathcal{G},\mathcal{G})>0$ for every $\mathcal{G}$ (for instance, there is a non-zero probability of all the nodes being inactive). As such, the corresponding stationary distribution, $\mu^{*}_m$, is a limiting distribution. Furthermore, due to Lemma \ref{SymReg3R}, $P^*_m$ is symmetric, and it is consequently doubly stochastic. As a result, $\mu^*_m$ is uniform over $\mathbb{G}^0_{n,m}$.

\end{proof}

\begin{theorem}
\label{limdist3R} For any connected $\mathcal{G}=(V,E,\bold{0})$ satisfying $\bar{d}(\mathcal{G})>2$, $P^*$ leads to a limiting distribution, $\mu^*_\mathcal{G}$, given as
\begin{equation}
\label{limdisteq}
\mu^*_\mathcal{G}= \sum_{m\in M} \Pr[\mathcal{G} \rightarrow \mathbb{G}^0_{n,m}]\mu^*_m,
\end{equation}
where $\Pr[\mathcal{G} \rightarrow \mathbb{G}^0_{n,m}]$ is the probability that the chain initialized at $\mathcal{G}$ ever enters the set of $m$-regular graphs, $\mathbb{G}^0_{n,m}$, and each $\mu^*_m$ is uniform over its support, $\mathbb{G}^0_{n,m}$.
\end{theorem}
\begin{proof}
Due to  Lemma \ref{reg_feas}, the Markov chain initialized at any connected $\mathcal{G}=(V,E,\bold{0})$ satisfying $\bar{d}(\mathcal{G})>2$ almost surely enters a closed communicating class of $m$-regular graphs, $\mathbb{G}^0_{n,m}$, such that $m\in M$. Furthermore, each $\mathbb{G}^0_{n,m}$ has a uniform limiting distribution, $\mu^*_m$ as given in Lemma \ref{limincom}. Hence, $P^*$ leads to the convex combination of limiting distributions, $\mu^*_m$, where each $\mu^*_m$ is weighted by the probability that the chain starting at $\mathcal{G}$ ever enters $\mathbb{G}^0_{n,m}$.
\end{proof} 

In light of Theorem \ref{limdist3R}, Algorithm \ref{RoundRandRegAlg} asymptotically transforms any connected initial graph, $\mathcal{G}=(V,E,\bold{0})$, satisfying $\bar{d}(\mathcal{G})>2$  into a random $m$-regular graph for some $m\in M$ as given in (\ref{mvals}). When $M$ is not a singleton, the probability of reaching each $\mathbb{G}^0_{n,m}$ depends on the initial graph as well as the algorithm parameters $\epsilon, \beta \in (0,1)$. However, the graphs observed in the limit are guaranteed to be expanders since $\bar{d}(\mathcal{G})>2$ implies $m\geq 3$ for every $m\in M$.


%
%
%

\section{Simulation Results}
\label{sims}
In this section, we present some simulation results for the proposed scheme. An arbitrary connected graph, $\mathcal{G}(0)=(V,E,\bold{0})$, is generated with $100$ nodes and $135$ edges. As such, the initial average degree is $\bar{d}(\mathcal{G}(0))=2.7$. The interaction graph is evolved via Algorithm \ref{RoundRandRegAlg} with $\epsilon = \beta =0.01$ for a period of 100000 steps. Since $\bar{d}(\mathcal{G}(0))=2.7$, the average degree is guaranteed to remain in $[2.7,4.7]$ for any feasible trajectory, as given in Corollary \ref{edgewcor}. Since the number of nodes is even, both integers in this interval are feasible values for the average degree, i.e. $M=\{3,4\}$. In light of Theorem \ref{limdist3R}, the graph is expected to become either a random $3$-regular graph or a random $4$-regular graph. In the presented simulation, the graph reaches a 3-regular graph after 45123 time steps. Note that once the system enters $\mathbb{G}^0_{100,3}$, both $f(\mathcal{G}(t))$ and $\bar{d}(\mathcal{G}(t))$ are stationary. Fig.
\ref{init2rr3R} illustrates $\mathcal{G}(0)$ and $\mathcal{G}(100000)$, which have the algebraic connectivities of 0.02 and 0.25,
respectively.  The values of the degree range, $f(\mathcal{G}(t))$, and the average degree, $\bar{d}(\mathcal{G}(t))$, for the first 50000 steps are illustrated in Fig. \ref{dbarrange3R}. Once a 3-regular graph is reached, the graph keeps randomizing in $\mathbb{G}^0_{100,3}$ via $r_2$ and $r_4$. The evolution of the algebraic connectivity throughout this simulation is shown in Fig. \ref{algcon3R}. As expected from random 3-regular graphs, the algebraic connectivity of the network is observed to be at least $3-2\sqrt{2}$ with a very high probability after a sufficient amount of time.

\begin{figure*}[H!]
\begin{center}
\includegraphics[trim = 0mm 60mm 0mm 60mm,clip,scale=0.4]{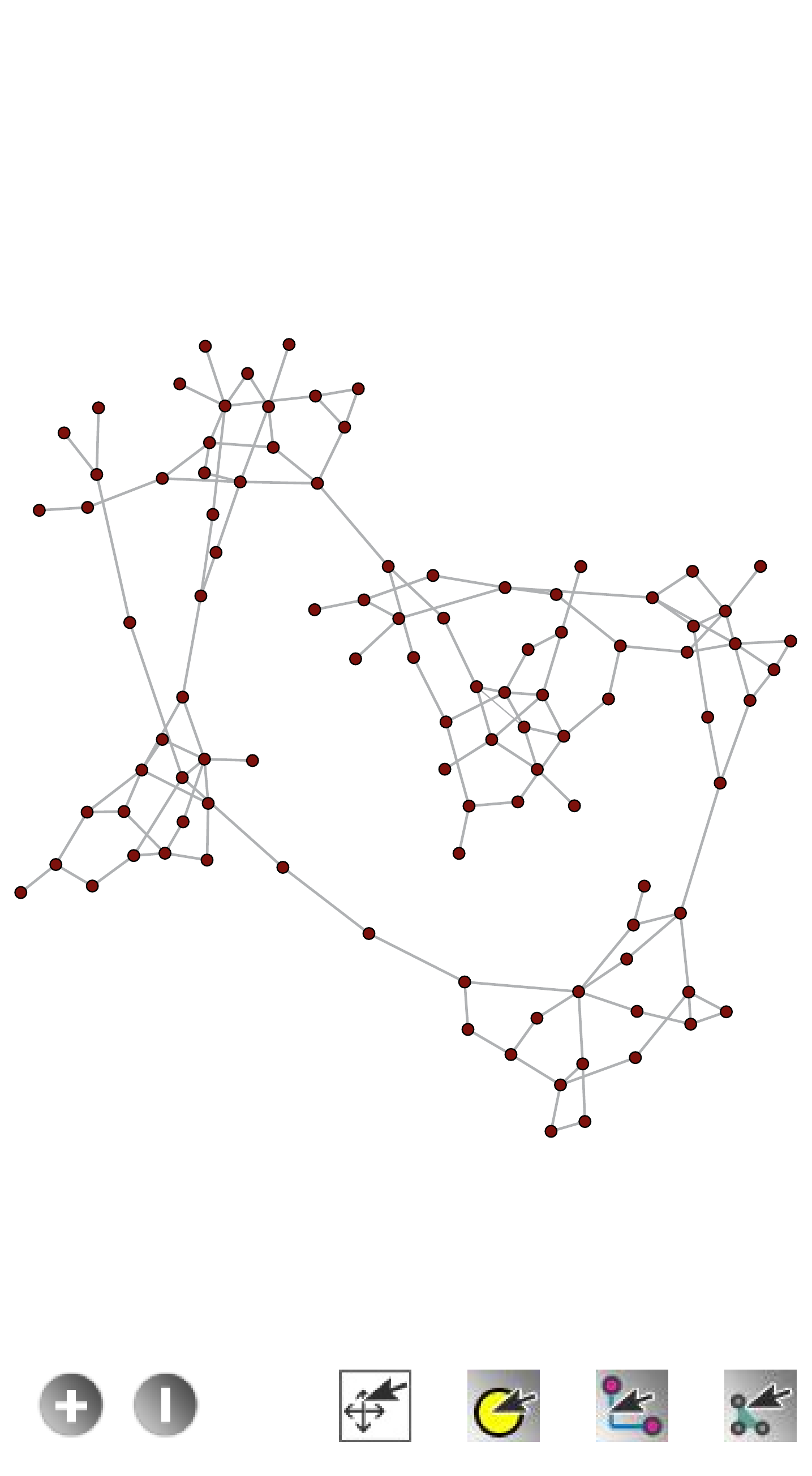}
\hspace{10mm}
\includegraphics[trim = 0mm 40mm 0mm 50mm,clip,scale=0.3]{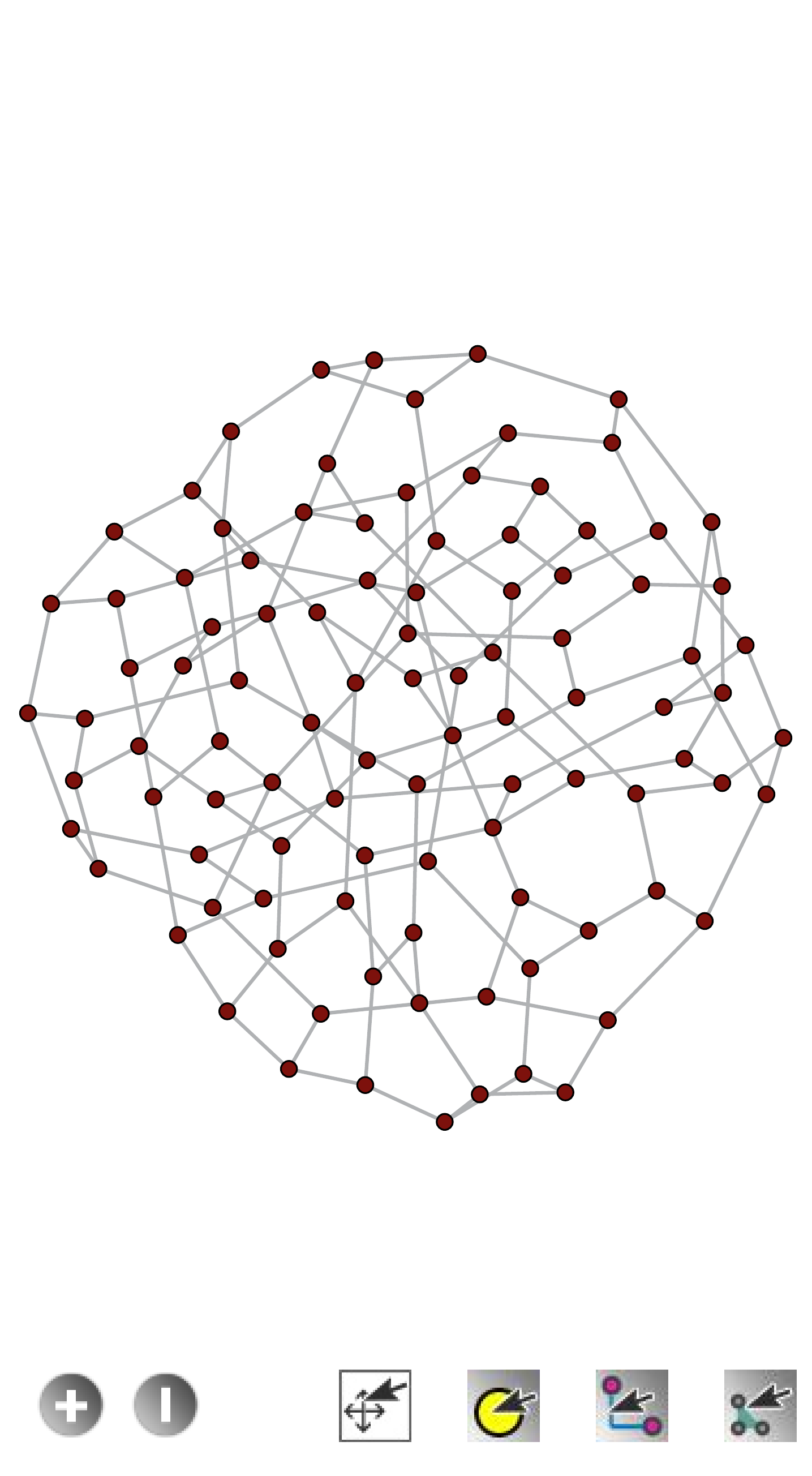}
\put(-260,-15){\large (a)}
\put(-70,-15){\large (b)}
\caption{Agents transform the initial interaction graph in (a), which has an average degree of 2.7, into a random 3-regular graph such as the one in (b) by following Algorithm \ref{RoundRandRegAlg}.  }
\label{init2rr3R}
\end{center}
\end{figure*}

\begin{figure*}[H!]
\begin{center}
\includegraphics[trim = 30mm 0mm 0mm 10mm,clip,scale=0.4]{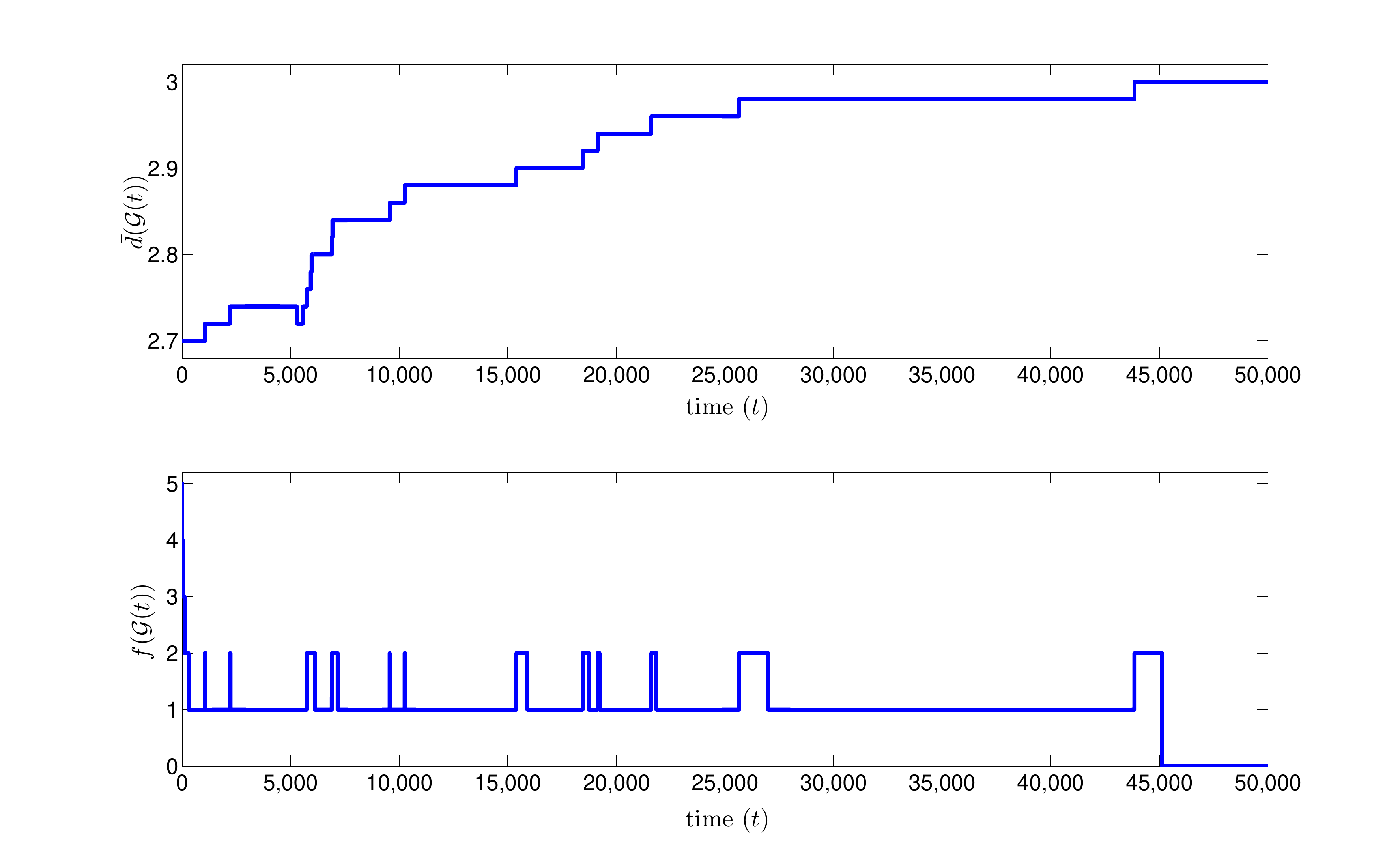}
\caption{The average degree, $\bar{d} (\mathcal{G}(t))$, and the degree range, $f (\mathcal{G}(t))$, for the first 50000 time steps of the simulation. Once a regular graph (i.e., $f (\mathcal{G}(t))=0$) is reached, both $\bar{d} (\mathcal{G}(t))$ and $f (\mathcal{G}(t))$ remain stationary under Algorithm \ref{RoundRandRegAlg}. 
}
\label{dbarrange3R}
\end{center}
\end{figure*}
\begin{figure*}[H!]
\begin{center}
\includegraphics[trim = 30mm 0mm 0mm 10mm,clip,scale=0.25]{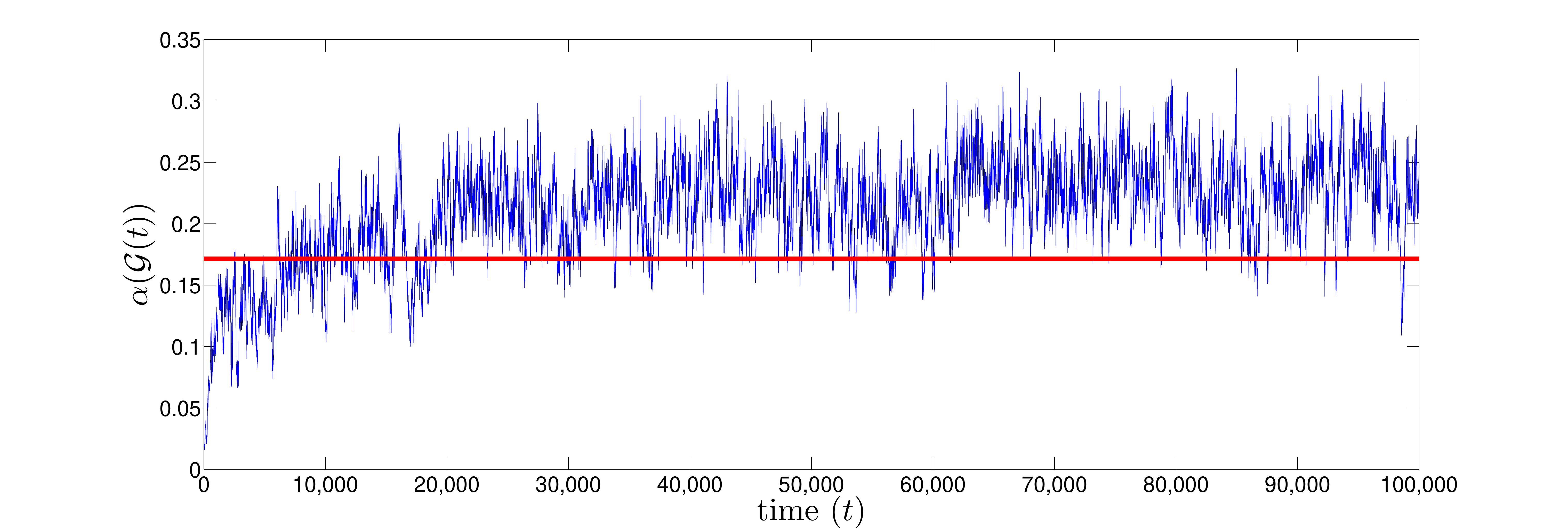}
\caption{The algebraic connectivity, $\alpha (\mathcal{G}(t))$, as the initial graph in Fig. \ref{init2rr3R}a evolves via Algorithm \ref{RoundRandRegAlg}. After a sufficiently large amount of time, $\alpha (\mathcal{G}(t))$ rarely drops below $3-2\sqrt{2}$ (marked with a horizontal solid line), as expected from random 3-regular graphs.
}
\label{algcon3R}
\end{center}
\end{figure*}
\section{Conclusions}
\label{conclusion}
In this paper, we presented a decentralized graph reconfiguration scheme for building robust multi-agent networks. In particular, we focused on the connectivity properties of the interaction graph as the robustness measure. Accordingly, we provided a decentralized method for transforming interaction graphs into well-connected graphs with a similar sparsity as the initial graph. More specifically, the proposed solution produces random $m$-regular graphs in the limit. Such graphs are expanders for any $m \geq 3$, i.e. they have the algebraic connectivity and expansion ratios bounded away from zero regardless of the network size. 



The proposed method was incrementally built in the paper. First, a single-rule grammar, $\Phi_R$, was designed for balancing the degree distribution in any connected network with an initial average degree $k$. $\Phi_R$ transforms the initial graph into a connected $k$-regular graph if $k \in \mathbb{N}$. If $k \notin \mathbb{N}$, then the degree range converges to 1 via $\Phi_R$. Next, $\Phi_R$ was extended to ensure that the interaction graph does not convergence to a poorly-connected $k$-regular graph. To this end, a local link randomization rule was added to $\Phi_R$. The resulting grammar, $\Phi_{RR}$, transforms the initial graph into a connected random $k$-regular graph if $k \in \mathbb{N}$.  Finally, $\Phi_{RR}$ was extended to obtain random regular graphs in the most general case, i.e. even when $k \notin \mathbb{N}$.  For any $k>2$, the resulting grammar, $\Phi^*$, leads to a connected random $m$-regular graph such that $m \in [k, k+2]$. Note that $k>2$ implies $m \geq 3$. Hence, the corresponding random $m$-regular graphs are expanders  with a similar sparsity as the initial graph. Some simulation results were also presented in the paper to demonstrate the performance of the proposed method.

\bibliography{MyReferences}

\end{document}